\documentclass[11pt,journal,draftcls,onecolumn,peerreviewca]{IEEEtran}
\usepackage{amsfonts}
\usepackage{amsmath}
\usepackage{amssymb}
\usepackage{bm}
\usepackage{xcolor,graphicx,float}
\usepackage{color, soul}
\usepackage{stfloats}
\usepackage[numbers,sort&compress]{natbib}
\usepackage[amsmath,thmmarks]{ntheorem}
\usepackage{theorem}
\usepackage{algorithm}
\usepackage{algorithmic}
\usepackage{graphicx}
\usepackage{subfigure}
\usepackage{pict2e}

\newtheorem{proposition}{Proposition}

\theoremheaderfont{\sc}\theorembodyfont{\upshape}
\theoremstyle{nonumberplain}
\theoremseparator{}
\theoremsymbol{\rule{1ex}{1ex}}
\newtheorem{proof}{Proof}

\begin{document}
%
\title{Newtonalized Orthogonal Matching Pursuit for Linear Frequency Modulated Pulse Frequency Agile Radar}
%
%

\author{Jiang Zhu, Honghui Guo, Ning Zhang, Chunyi Song and Zhiwei Xu
\thanks{J. Zhu, H. Guo, C. Song and Z. Xu are with the Ocean College, Zhejiang University, Zhoushan, 316021, China (e-mail: \{jiangzhu16, honghuiguo, cysong, xuzw\}@zju.edu.cn).}
\thanks{N. Zhang is with Nanjing Marine Radar Institute, Jiangning District, Nanjing, 211153, China (e-mail: zhangn\_ee@163.com).}}

\maketitle

\begin{abstract}
The linear frequency modulated (LFM) frequency agile radar (FAR) can synthesize a wide signal bandwidth through coherent processing while keeping the bandwidth of each pulse narrow. In this way, high range resolution profiles (HRRP) can be obtained without increasing the hardware system cost. Furthermore, the agility provides improved both robustness to jamming and spectrum efficiency. Motivated by the Newtonalized orthogonal matching pursuit (NOMP) for line spectral estimation problem, the NOMP for the FAR radar termed as NOMP-FAR is designed to process each coarse range bin to extract the HRRP and velocities of multiple targets, including the guide for determining the oversampling factor and the stopping criterion. In addition, it is shown that the target will cause false alarm in the nearby coarse range bins, a postprocessing algorithm is then proposed to suppress the ghost targets. Numerical simulations are conducted to demonstrate the effectiveness of NOMP-FAR.
\end{abstract}


\IEEEpeerreviewmaketitle
\section{Introduction}
Linear stepped frequency radar (LSFR) in which the frequency is stepped linearly with a constant frequency change is used in many remote-sensing related applications, such as airborne synthetic aperture radar (SAR) \cite{MensaHRbook, Soumekhbook}, inverse SAR \cite{LaneISAR} and ground penetrating radar \cite{GPR1999, Davis}. For the LSFR, the transmitted signal bandwidth is increased sequentially during the coherent processing interval (CPI), yielding high range resolution \cite{SFW1984}. In detail, the carrier frequency of successive pulses is increased linearly. By coherently processing all pulses together, the LSFR indirectly synthesizes a wide bandwidth and keeps the bandwidth of each pulse low. This significantly decreases the analog-to-digital converter (ADC) sampling rate and reduces the complexity of the receiver. Meanwhile, the high-range resolution profile (HRRP) can be obtained.

However, the LSFR has its own drawbacks. The LFSR suffers from poor antijamming, especially in electronic counter-countermeasures (ECCM) \cite{GarNarTAES2002}. Since the carrier frequency variation is linear, an interceptor can easily track and predict the transmitted waveform, and then implement effective jamming. In contrast, the transmitted waveform of the frequency agile radar (FAR) where the carrier frequencies are varied in a random manner from pulse to pulse are difficult to track and predict, see Fig. \ref{SFWFAR} for an illustration, and thus FAR has excellent ECCM performance \cite{HuangTAES, RassTer}. Besides, FAR suppresses the range ambiguity and improves covert detection. In addition, FAR exploits vacant spectral bands flexibly, which increases the spectrum efficiency and can be a more effective solution in a contested, congested and competitive electromagnetic environment. Finally, FAR can simultaneously convey messages by embedding the information into the frequency hopping pattern \cite{HuangTSP2020MCAPAR, MaSPM2020}, which utilizes the spectrum more efficiently and reduces hardware cost.

The advantages of the FAR over LSFR has made FAR receive increasing attention in the radar community. In \cite{AxelssonTGRS2007}, the statistical characteristics of the ambiguity function and the sidelobe noise floor are analyzed. Later, a range velocity iterative alternating projection (RV-IAP) algorithm for target range-velocity estimation is proposed to avoid the sidelobe pedestal problem in the multiple targets scenario \cite{LiuEL2008}. The RV-IAP is based on the orthogonal matching pursuit (OMP) method \cite{OMPmethod}. In \cite{HuangTSP2018}, theoretical analysis of FAR via compressed sensing (CS) algorithms is studied, in particular the properties of the sensing matrix of the FAR. The bounds on the number of recoverable targets is derived, with the assumption that the targets are on the grid exactly. Later, theoretical analysis of FAR for extended targets has been studied \cite{WangarXiv2019}. When the targets are not on the grid, model mismatch occurs and the on-grid based algorithms may degrade significantly \cite{ChiTSP2011}.
\begin{figure}
  \centering
  \includegraphics[width=8cm]{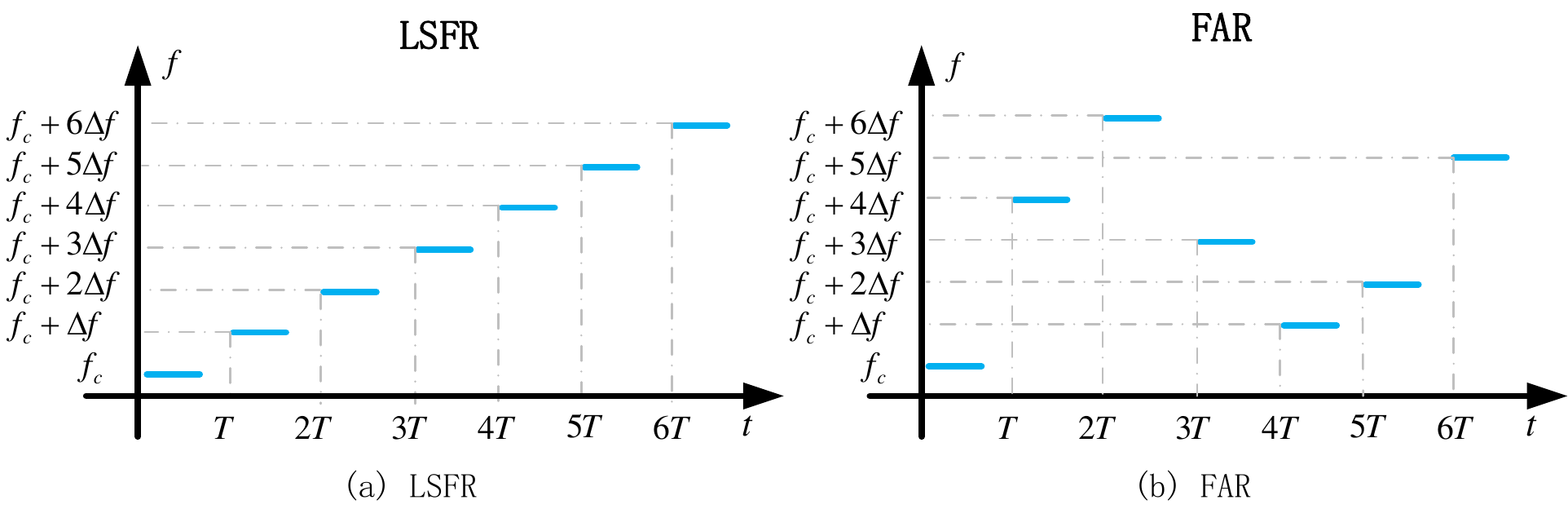}\\
  \caption{The carrier frequency versus time. (a) SFWR, (b) FAR.}\label{SFWFAR}
\end{figure}

This paper studies the linear frequency modulated (LFM) pulse FAR, in contrast with the previous monopulse radar signal model \cite{HuangTSP2018}. By referring to the Newtonalized OMP (NOMP) algorithm for line spectral estimation and the other communication and radar related applications \cite{NOMPTSP2016, MNOMPSP2019, CSSP2016, Gupta2016, Jinshi}, we design NOMP for the FAR termed as NOMP-FAR. NOMP-FAR avoids the grid mismatch problem by using a Newton refinement step and a feedback strategy to improve the estimation of already detected targets. A stopping criterion based on constant false alarm rate (CFAR) is provided and the model order is determined. Compared to \cite{HuangTSP2018} which adopts the Nyquist grid\footnote{The definition of Nyquist grid is borrowed from \cite{HuangTSP2018} and also defined later after eq. (\ref{defSpq}) in the Coarse detection step.}, it is shown that oversampling is needed to ensure the convergence of the Newston step. Besides, it is shown that targets will cause false alarms in the nearby range bins, and the relationship between the target range and velocity in the true range bin and that in the nearby range bins is revealed. This relationship is further utilized to identify and suppress the false alarms. Finally, substantial numerical experiments are conducted to both demonstrate the effectiveness of the proposed NOMP-FAR and corroborate the theoretical analysis.
%

\section{System Model}
Consider the FAR system that transmits a total of $N$ pulses in a coherent processing interval (CPI) with a constant pulse repetition interval (PRI) $T$. The transmitted LFM signal is
\begin{align}\label{LFMpulse}
s(t)={\rm rect}\left(\frac{t}{T_p}\right){\rm e}^{{\rm j}\kappa\pi (t-\frac{T_p}{2})^2},\quad 0\leq t\leq T_p,
\end{align}
where $T_p$ denotes the duration of each pulse and $\kappa$ denotes the chirp rate. For the above signal (\ref{LFMpulse}), its bandwidth is $B=\kappa T_p$. Let $B_c$ denote the total transmit bandwidth available at the baseband. For the $n$th pulse, the carrier frequency is $f_n$ satisfying $f_n\in [f_c,f_c+B_c]$, $n=0,\cdots,N-1$, where $f_c$ denotes the minimal carrier frequency. The frequencies $\{f_n\}_{n=0}^{N-1}$ are drawn uniformly at random from the set ${\mathcal F}=\{f_n|f_n=f_c+d_n\Delta f,~d_n\in \{0,1,2,\cdots,M-1\}\subset {\mathbb Z}\}$, where $M=\lfloor B_c/\Delta f\rfloor+1$ and $d_n$ refers to the $n$th random frequency-modulation code.

To develop the system model, let $nT$ denote the time origin. For simplicity, a single target echo model is established. Then the multiple target model is generalized. The $n$th transmitted pulse is
\begin{align}
x_n(t)=s(t){\rm e}^{{\rm j}2\pi f_nt}.
\end{align}
Besides, at the $t$th instant, the target position $r(t)$ is
\begin{align}
r(t)=r(0)+\nu\times (nT+t),
\end{align}
where $r(0)$ denotes the original position in the $1$st pulse and $\nu$ denotes the velocity. The $n$th received echo signal $r_n(t)$ is
\begin{align}
   r_n(t)&=\beta x\left( t-\frac{2r(t)}{c} \right)+w_n(t)\notag\\
   &\stackrel{a}\approx \beta s\left( t-\frac{2(r(0)+\nu nT)}{c} \right){{\text{e}}^{\text{j}2\pi {{f}_{n}}\left( t-\frac{2(r(0)+\nu (t+nT))}{c} \right)}}+w_n(t) \notag\\
 &=\beta s\left( t-{{\tau }_{0}}-\frac{2\nu nT}{c} \right){{\text{e}}^{\text{j}2\pi {{f}_{n}}(t-{{\tau }_{0}})}}{{\text{e}}^{\text{j}2\pi {{f}_{dn}}(t+nT)}}+w_n(t),\notag
\end{align}
where $w_n(t)$ is the white Gaussian noise and $w_n(kT_s)\sim {\mathcal {CN}}(0,\sigma^2)$ with $T_s$ being the sampling interval, $\stackrel{a}\approx$ is due to the ``stop-and-hop'' assumptions \cite{FRSP2005},
\begin{align}
\tau_0=\frac{2r(0)}{c}
\end{align}
denotes the echo delay in the first pulse,
\begin{align}
{f}_{dn}\triangleq \frac{2\nu}{c}f_n
\end{align}
denotes the Doppler shift. The reference signal for the $n$th pulse is
\begin{align}\label{refsig}
x_{{\rm ref},n}(t)=\frac{1}{\sqrt{N_{\rm ref}}}x_n^*(-t)=\frac{1}{\sqrt{N_{\rm ref}}}s^*(-t){\rm e}^{-{\rm j}2\pi f_n(-t)},
\end{align}
where $N_{\rm ref}$ denotes the number of sampling points and ${1}/{\sqrt{N_{\rm ref}}}$ is to ensure that $\sum\limits_{k=0}^{N_{\rm ref}-1}|x_{{\rm ref},n}(kT_s)|^2=1$. Pulse compression (PC) is conducted on the echo signal to obtain
\begin{align}
  y_n(t)&=\int_{-\infty }^{\infty }{r}_n(\tau ){x_{{\rm ref},n}}(t-\tau )\text{d}\tau \notag\\
 &=\frac{\beta}{\sqrt{N_{\rm ref}}}{{\text{e}}^{\text{j}2\pi {{f}_{n}}(t-{{\tau }_{0}})}}\int_{-\infty }^{\infty }{{{\text{e}}^{\text{j}2\pi {{f}_{dn}}(\tau +nT)}}}s(\tau -{{\tau }_{0}}-\frac{2\nu nT}{c}){{s}^{*}}(\tau -t)\text{d}\tau+\tilde{w}_n(t),
\end{align}
where $\tilde{w}_n(t)\triangleq \frac{1}{\sqrt{N_{\rm ref}}}\int_{-\infty }^{\infty }w_n(\tau){{s}^{*}}(\tau -t){{\text{e}}^{-\text{j}2\pi {{f}_{n}}(\tau -t)}}\text{d}\tau$. Note that the variance of $\tilde{w}_n(kT_s) $ is still $\sigma^2$ due to the normalizing property $\sum\limits_{k=0}^{N_{\rm ref}-1}|x_{{\rm ref},n}(kT_s)|^2=1$. Let $\tau-t$ be replaced with $u$ and the integral domain be revised accordingly, then $y_n(t)$ is rewritten as
\begin{align}
  y_n(t)=\frac{\beta}{\sqrt{N_{\rm ref}}} {{\text{e}}^{\text{j}2\pi {{f}_{n}}(t-{{\tau }_{0}})}}{{\text{e}}^{\text{j}2\pi {{f}_{dn}}(t+nT)}} \chi (-t+{{\tau }_{0}}+\frac{2\nu nT}{c},{{f}_{dn}})+\tilde{w}_n(t),\notag
\end{align}
where ${\rm \chi}(\xi,f_d)$ denotes the ambiguity function defined as
\begin{align}\label{amfdef}
{\rm \chi}(\xi,f_d)=\int_{-\infty}^{\infty}{\rm e}^{{\rm j}2\pi f_dt}s(t-\xi)s^*(t){\rm d}t.
\end{align}
For the LFM signal, its ambiguity function is
\begin{equation}
	\label{amfun}
	{\rm \chi}(\xi,f_d)=\begin{cases}{\rm e}^{{\rm j}\pi f_d(T_p+\xi)}\left(T_p-|\xi|\right){\rm sinc}[(f_d-\kappa \xi)(T_p-|\xi|)],\quad |\xi|<T_p,\\
0,\quad {\rm otherwise},
\end{cases}
	\end{equation}
where ${\rm sinc}(x) = \frac{\sin(\pi x)}{\pi x}$. Note that $|{\rm \chi}(\xi,f_d)|\leq {\rm \chi}(0,0)=T_p$.

The $n$th echo signal is sampled at the rate $F_s=1/T_s$. Assuming the target lies in the $l$th coarse range cell corresponding to the sampling instant $t_l=(l-1)T_s$, $l=\lfloor\frac{R_{\rm min}}{c/(2F_s)}+\frac{1}{2}\rfloor+1,\cdots,\lfloor\frac{R_{\rm max}}{c/(2F_s)}+\frac{1}{2}\rfloor+1$, where $R_{\rm min}$ and $R_{\rm max}$ denotes the range interval of interest. Let
\begin{align}\label{HRRPdef}
R=(\tau_0-t_l)c/2
\end{align}
denote the high resolution range profile. The sampled, discrete received signal for the $n$th pulse is
\begin{align}
  & y_n(t_l)=\frac{\beta}{\sqrt{N_{\rm ref}}} {{\text{e}}^{\text{j}2\pi {{f}_{n}}({t_l}-{{\tau }_{0}})}}{{\text{e}}^{\text{j}2\pi {{f}_{dn}}({t_l}+nT)}} \chi \left(-t_l+{{\tau }_{0}}+\frac{2\nu nT}{c},{{f}_{dn}}\right)+\tilde{w}_n(t_l) \notag\\
 &\stackrel{b}\approx \left( \frac{\beta}{\sqrt{N_{\rm ref}}} {{\text{e}}^{\text{j}2\pi {{f}_{c}}(-\frac{2\nu }{c}t_l)}}{{\text{e}}^{-\text{j}\frac{4\pi {{f}_{c}}R}{c}}} \right){{\text{e}}^{-\text{j}\frac{4\pi \Delta fR}{c}{{d}_{n}}}}{{\text{e}}^{-\text{j}\frac{4\pi T\nu {{f}_{c}}}{c}(1+{{d}_{n}}\Delta f/{{f}_{c}})n}}\chi \left(\frac{2R}{c}+\frac{2\nu nT}{c},{{f}_{dn}}\right)+\tilde{w}_n(t_l),\label{form9}
\end{align}
where $\stackrel{b}\approx$ is due to the approximation ${{\text{e}}^{\text{j}2\pi {{f}_{n}}(-\frac{2\nu }{c}t_l)}}\approx {{\text{e}}^{\text{j}2\pi {{f}_{c}}(-\frac{2\nu }{c}t_l)}}$. Define the maximum unambiguous high-resolution range and unambiguous velocity as
\begin{align}\label{rangevel}
{{R}_{u}}=\frac{c}{2\Delta f},\quad {{\nu }_{u}}=\frac{c}{2{{f}_{c}}T}.
\end{align}
Substituting (\ref{rangevel}) in (\ref{form9}), the sampled, discrete received signal $y(n,t_l)$ can be further formulated as
\begin{align}
   y_n(t_l)&\approx \left( \frac{\beta}{\sqrt{N_{\rm ref}}} {{\text{e}}^{\text{j}2\pi {{f}_{c}}(-\frac{2\nu }{c}t_l)}}{{\text{e}}^{-\text{j}\frac{4\pi {{f}_{c}}R}{c}}} \right){{\text{e}}^{-\text{j}\frac{4\pi \Delta fR}{c}{{d}_{n}}}}{{\text{e}}^{-\text{j}\frac{4\pi T\nu {{f}_{c}}}{c}(1+{{d}_{n}}\Delta f/{{f}_{c}})n}}\chi \left(\frac{2R}{c}+\frac{2\nu nT}{c},{{f}_{dn}}\right)+\tilde{w}_n(t_l) \notag\\
 &= \underbrace{\tilde{\gamma} {{\text{e}}^{\text{j}p{{d}_{n}}+\text{j}q(1+{{d}_{n}}\Delta f/{{f}_{c}})n}}\chi \left( -\frac{p}{2\pi \Delta f}-\frac{qn}{2\pi {{f}_{c}}},\frac{q}{2\pi T}(1+{{d}_{n}}\Delta f/{{f}_{c}}) \right)}_{z_n(t_l)}+\tilde{w}_n(t_l),\label{form12}
\end{align}
where
\begin{align}
  & \tilde{\gamma} =\frac{\beta}{\sqrt{N_{\rm ref}}} {{\text{e}}^{\text{j}2\pi {{f}_{c}}(-\frac{2\nu }{c}t_l)}}{{\text{e}}^{-\text{j}\frac{4\pi {{f}_{c}}R}{c}}}, \\
 & p=-2\pi \frac{R}{{{R}_{u}}}=-\frac{4\pi \Delta f}{c}R, \\
 & q=-2\pi \frac{\nu }{{{\nu }_{u}}}=-\frac{4\pi {{f}_{c}}T}{c}\nu .
\end{align}

It is worth noting that (\ref{form12}) is very general for arbitrary waveforms by replacing the corresponding ambiguity function $\chi(\cdot,\cdot)$.
Taking the noise and multiple targets into account, (\ref{form12}) can be formulated as the following mathematical problem
\begin{align}\label{modelam}
{{y}_{n}}=\sum\limits_{k=1}^{K}{{\tilde{\gamma}_{k}}}{{\text{e}}^{\text{j}{{p}_{k}}{{d}_{n}}+\text{j}{{q}_{k}}(1+{{d}_{n}}\Delta f/{{f}_{c}})n}} \chi \left( -\frac{{{p}_{k}}}{2\pi \Delta f}-\frac{{{q}_{k}}n}{2\pi {{f}_{c}}},\frac{{{q}_{k}}}{2\pi T}(1+{{d}_{n}}\Delta f/{{f}_{c}}) \right)+{{w}_{n}},
\end{align}
where $n=0,\cdots ,N-1$, $w_n$ is independent and identically distributed and $w_n\sim {\mathcal {CN}}(0,\sigma^2)$. Without loss of generality, the noise variance is set to be $\sigma^2=1$. From the signal processing perspective, the goal is to recover the digital frequencies $p_k\in [-\pi,\pi)$ and $q_k\in [-\pi,\pi)$ from the noisy measurements ${\mathbf y}=[y_0,y_1,\cdots,y_{N-1}]^{\rm T}$. Once the digital frequency is recovered, the range and velocity estimates are
\begin{align}\label{ptor}
\hat{r}=-\frac{c}{4\pi\Delta f}p+\frac{ct_l}{2},
\end{align}
and
\begin{align}\label{qtov}
\hat{v}=-\frac{c}{4\pi f_c T}q.
\end{align}

Note that $\chi \left( -\frac{{{p}_{k}}}{2\pi \Delta f}-\frac{{{q}_{k}}n}{2\pi {{f}_{c}}},\frac{{{q}_{k}}}{2\pi T}(1+{{d}_{n}}\Delta f/{{f}_{c}}) \right)\approx const$ for typical system parameters, $\forall n$, which can be absorbed into $\tilde{\gamma}_k$. Define ${\gamma}_k\triangleq\sqrt{N}\tilde{\gamma}_k\chi \left( -\frac{{{p}_{k}}}{2\pi \Delta f}-\frac{{{q}_{k}}n}{2\pi {{f}_{c}}},\frac{{{q}_{k}}}{2\pi T}(1+{{d}_{n}}\Delta f/{{f}_{c}}) \right)$. Consequently, model (\ref{modelam}) can be further approximated and simplified as
\begin{align}\label{modelfinal}
{{y}_{n}}=\sum\limits_{k=1}^{K}{{{{{\gamma }}}_{k}}}\frac{1}{\sqrt{N}}{{\text{e}}^{\text{j}{{p}_{k}}{{d}_{n}}+\text{j}{{q}_{k}}(1+{{d}_{n}}\Delta f/{{f}_{c}})n}}+{{w}_{n}},
\end{align}
where $n=0,\cdots ,N-1$. This is consistent with \cite{HuangTSP2018} studying the FAR system using monotone pulses. In the following section, the NOMP-FAR is developed for solving problem (\ref{modelfinal}).

%

\section{NOMP-FAR Algorithm}
At first, the NOMP-FAR algorithm is developed for a single target scenario. Then it is generalized to multiple target scenarios \cite{NOMPTSP2016}. We present the details in the following.

The model (\ref{modelfinal}) is formulated as a vector form
\begin{align}\label{modelfinalvector}
{\mathbf y}=\sum\limits_{k=1}^K\gamma_k{\mathbf a}(p_k,q_k)+{\mathbf w},
\end{align}
where
${\mathbf w}\sim {\mathcal {CN}}(0,\sigma_w^2{\mathbf I})$, ${\mathbf a}(p,q)$ is
\begin{align}
\mathbf{a}(p,q)=\frac{1}{\sqrt{N}}\left[ \begin{matrix}
   {{\text{e}}^{\text{j}p{{d}_{0}}}}  \\
   \vdots   \\
   {{\text{e}}^{\text{j}p{{d}_{n}}}}{{\text{e}}^{\text{j}q(1+{{d}_{n}}\Delta f/{{f}_{c}})n}}  \\
   \vdots   \\
   {{\text{e}}^{\text{j}p{{d}_{N-1}}}}{{\text{e}}^{\text{j}q(1+{{d}_{N-1}}\Delta f/{{f}_{c}})(N-1)}}  \\
\end{matrix} \right]
\end{align}
and $\|\mathbf{a}(p,q)\|_2=1$.

\subsection{A Single Target Scenario}
Suppose that there is only a single target, i.e.,
\begin{align}
{\mathbf y}=\gamma{\mathbf a}(p,q)+{\mathbf w}.
\end{align}
The maximum likelihood (ML) estimation can be performed to yield the least squares problem
\begin{align}\label{cost1}
\underset{p,q,\gamma}{\operatorname{minimize}}~\|{\mathbf y}-\gamma{\mathbf a}(p,q)\|^2.
\end{align}
Given $p$ and $q$, the optimal $\gamma$ is
\begin{align}
\gamma_{\rm opt}=\frac{{\mathbf a}^{\rm H}(p,q){\mathbf y}}{\|{\mathbf a}(p,q)\|_2^2}={{\mathbf a}^{\rm H}(p,q){\mathbf y}}.
\end{align}
Eliminating $\gamma$, problem (\ref{cost1}) can be further simplified as
\begin{align}\label{defSpq}
\underset{p,q}{\operatorname{maximize}}~\|{\mathbf y}^{\rm H}{\mathbf a}(p,q)\|^2\triangleq S_w(p,q).
\end{align}

To solve the problem, a two stage procedure named coarse detection and refinement is adopted \cite{NOMPTSP2016}.
\begin{itemize}
  \item Coarse detection: By restricting $p$ and $q$ to the finite discrete set denoted by $\Omega_p \triangleq \{k_p(2\pi/{\gamma_p}M)-\pi:k_p=0, 1, \cdots, ({\gamma}_pM-1)\}$ and $\Omega_q \triangleq \{k_q(2\pi/{\gamma_q}N)-\pi:k_q=0, 1, \cdots, ({\gamma}_qN-1)\}$, where $\gamma_p$ and $\gamma_q$  are the oversampling factors relative to the Nyquist grid \cite{HuangTSP2018}, we can obtain a coarse estimate of $(p,q)$. We treat the $p_c\in{\Omega_p}$ and $q_c\in{\Omega_q}$ that maximize the cost function (\ref{defSpq}) as the output of this stage, and the corresponding $\gamma$ estimate is ${\gamma}_c={\mathbf a}^{\rm H}(p_c,q_c){\mathbf y}/\Vert {{\mathbf a}(p_c,q_c)} \Vert^2_{\rm{2}}$.
  \item Refinement: Let $({\gamma}_c, p_c,q_c)$ denote the current estimate, then the Newton procedure for frequency refinement is
   \begin{equation}
	\label{rule}
  \left[
  \begin{array}{c}
    \hat{p} \\
    \hat{q}
  \end{array}\right]=\left[
  \begin{array}{c}
    p_c \\
    q_c
  \end{array}\right]-\left[
    \begin{array}{cc}
      \frac{\partial^2 S_w(p,q)}{\partial p^2} & \frac{\partial^2 S_w(p,q)}{\partial p\partial q} \\
      \frac{\partial^2 S_w(p,q)}{\partial q\partial p} & \frac{\partial^2 S_w(p,q)}{\partial q^2} \\
    \end{array}
  \right]^{-1}\left[\begin{array}{c}
    \frac{\partial S_w(p,q)}{\partial p} \\
    \frac{\partial S_w(p,q)}{\partial q}
  \end{array}\right] \Bigg|_{(p,q)=(p_c,q_c)}.
	\end{equation}
   Since $S_w(p,q)=\|{\mathbf y}^{\rm H}{\mathbf a}(p,q)\|^2={\mathbf y}^{\rm H}{\mathbf a}(p,q){\mathbf a}(p,q)^{\rm H}{\mathbf y}$, the elements of the gradient and Hessian matrix are given in
      \begin{subequations}
\begin{align}
\frac{\partial S_w(p,q)}{\partial p}&=2\Re\left\{{\mathbf y}^{\rm H}\frac{\partial {\mathbf a}(p,q)}{\partial p}{\mathbf a}(p,q)^{\rm H}{\mathbf y}\right\},\\
\frac{\partial S_w(p,q)}{\partial q}&=2\Re\left\{{\mathbf y}^{\rm H}\frac{\partial {\mathbf a}(p,q)}{\partial q}{\mathbf a}(p,q)^{\rm H}{\mathbf y}\right\},\\
\frac{\partial^2 S_w(p,q)}{\partial p^2}&=2\Re\left\{{\mathbf y}^{\rm H}\frac{\partial^2 {\mathbf a}(p,q)}{\partial p^2}{\mathbf a}(p,q)^{\rm H}{\mathbf y}\right\}+2\Re\left\{{\mathbf y}^{\rm H}\frac{\partial {\mathbf a}(p,q)}{\partial p}\frac{\partial {\mathbf a}^{\rm H}(p,q)}{\partial p}{\mathbf y}\right\},\\
\frac{\partial^2 S_w(p,q)}{\partial p\partial q}&=2\Re\left\{{\mathbf y}^{\rm H}\frac{\partial^2 {\mathbf a}(p,q)}{\partial p\partial q}{\mathbf a}(p,q)^{\rm H}{\mathbf y}\right\}+2\Re\left\{{\mathbf y}^{\rm H}\frac{\partial {\mathbf a}(p,q)}{\partial p}\frac{\partial {\mathbf a}^{\rm H}(p,q)}{\partial q}{\mathbf y}\right\},\\
\frac{\partial^2 S_w(p,q)}{\partial q^2}&=2\Re\left\{{\mathbf y}^{\rm H}\frac{\partial^2 {\mathbf a}(p,q)}{\partial q^2}{\mathbf a}(p,q)^{\rm H}{\mathbf y}\right\}+2\Re\left\{{\mathbf y}^{\rm H}\frac{\partial {\mathbf a}(p,q)}{\partial q}\frac{\partial {\mathbf a}^{\rm H}(p,q)}{\partial q}{\mathbf y}\right\}.
\end{align}
\end{subequations}
We maximize $S_w(p,q)$ by employing the update rule (\ref{rule}) on the condition that the function is locally concave. Let $R_s$ denote the number of Newton steps.
\end{itemize}

In Section \ref{OFD}, the oversampling factor is investigated.

%
%
%
\subsection{Multiple Targets Scenario}
Assume that we have already detected $L$ targets, and let ${\mathcal P}=\{(\gamma_l,p_l,q_l),l=1,\cdots,L\}$ denote the set of estimates of the detected targets. The residual measurement corresponding to this estimate is
\begin{align}
{\mathbf y}_r(P)={\mathbf y}-{\sum\limits_{l=1}^L \gamma_l{\mathbf a}(p_l,q_l)}.
\end{align}
The method of estimating multiple targets proceeds by employing the single target procedure to perform Newtonized coordinate descent on the residual energy $\Vert{{\mathbf y}_r(P)}\Vert_2^{2}$. One step of this coordinate descent involves adjusting all $(p_l,q_l)$. The procedure to refine the $l$th target is as follows: ${\mathbf y}_r(P\backslash \{{\gamma}_l, p_l, q_l\})$ now is referred to as the measurement ${\mathbf y}$ and the single target update step is utilised to refine $({\gamma}_l, p_l, q_l)$.

Refinement Acceptance Condition (RAC): This refinement step is accepted when it results in a strict improvement in $G_{{\mathbf y}_r(P\backslash \{\gamma_l, p_l, q_l\})}(p,q)$, namely, $G_{{\mathbf y}_r}(\hat{p},\hat{q})>G_{{\mathbf y}_r}(p_l,q_l)$. By doing this, we can make sure that the adopted refinement must decrease the overall residual energy.

In summary, firstly, we detect a frequency pair $(p,q)$ over the discrete set $\Omega_p$ and $\Omega_q$ by maximizing the cost function (\ref{defSpq}). Then we use the knowledge of the first-order and second-order derivative of the cost function to refine the estimate of $(p,q)$. Next, we use the information of all the other previously detected targets to further improve the estimation performance of every previously detected target one at a time. This step improves the accuracy of the algorithm and the rounds of Cyclic Refinement is $R_c$. Finally, we update $\boldsymbol \gamma$ by least squares methods. The NOMP-FAR is summarized as Algorithm \ref{NOMP-FAR}.

\begin{algorithm}[ht]
\caption{NOMP-FAR.}\label{NOMP-FAR}
1: \textbf{Procedure} EXTRACTTARGETS $({\mathbf y, \tau}):$\\
2: $m\leftarrow 0$, ${P}_0 = \{\}$\\
3: $\textbf{while}$ {$\underset{p \in {\Omega}_p, q \in {\Omega}_q}{\rm {max}}G_{{\mathbf y}_r({P}_m)}(p,q)>\tau$}\\
4: $m\leftarrow m+1$\\
5: $\textbf{IDENTIFY}$\\
$(p_c,q_c) = \underset{p \in {\Omega}_p, q \in {\Omega}_q}{\rm {argmax}} G_{{\mathbf y}_r({P}_{m-1})}(p,q)$, ${\gamma}_c = {\mathbf a}^{\rm H}(p_c, q_c){\mathbf y}_r({P}_{m-1})/\Vert {{\mathbf a}(p_c,q_c)} \Vert^2_{\rm{2}}$\\
\\6: ${P_m'}\leftarrow {P}_{m-1}\cup \{(p_c, q_c, \gamma_c)\}$
\\7: SINGLE REFINEMENT: Refine $(p_c, q_c)$ using single frequency Newton update algorithm (${R}_s$ Newton steps) to obtain improved estimates $(\hat{p}, \hat{q})$, and its corresponding $\gamma$ estimate is $\hat {\gamma} \leftarrow {\mathbf a}^{\rm H}(\hat p, \hat q){\mathbf y}_r({P}_{m-1})/\Vert {{\mathbf a}(\hat p, \hat q)} \Vert^2_{\rm{2}}$.
\\8: ${P_m''}\leftarrow {P}_{m-1}\cup \{(\hat {p}, {\hat q}, \hat{\gamma})\}$
\\9: CYCLIC REFINEMENT: Refine parameters in ${P}''_m$ one at a time: For each $(\hat{p},\hat{q},\hat{\gamma})\in{P_m''}$, we treat ${\mathbf y}_r({P_m''} \backslash \{(\gamma,p,q)\})$ as the measurement $\mathbf y$, and apply single frequency Newton update algorithm. We perform ${R}_c$ rounds of cyclic refinements. Let ${P_m'''}$ denote the new set of parameters.
\\10: UPDATE all ${\boldsymbol \gamma}$ estimate in ${P_m'''}$ by least squares: $\mathbf A \triangleq [{\mathbf a}({p_1},{q_1}),\cdots,{\mathbf a}({p_m},{q_m})]$, where $\{(p_i,q_i)\}_{i=1}^m$ are the frequencies in ${P_m'''}$. And ${\boldsymbol\gamma} = {\mathbf A}^\dagger {\mathbf y}$. \\Let ${P}_m$ denote the new set of parameters.\\
11: $\textbf{end while}$
\\12: \textbf{return} $P_m$
\label{code:recentEnd}
\end{algorithm}

\subsection{Complexity Analysis}
The computation complexity of the main steps of NOMP-FAR assuming that the algorithm runs for exactly $K$ iterations. Both checking whether the stopping criterion is satisfied and the detection step are implemented with computation complexity $O(\gamma_M\gamma_NMN^2)$. The Single Refinement Step takes $O(R_sN)$ operations per target, hence the total cost for Single Refinement Step is $O(R_sKN)$ operations. The Cyclic Refinement Step has complexity $O(R_cR_sK^2N)$. For the Update Step which involves computing the least squares solution, the computation complexity is $O(NK^2+K^3)$ per iteration if the pseudo-inverse is computed directly. The overall cost for the Update Step is  $O(NK^3+K^4)$. To summarize, the whole computation complexity of the NOMP-FAR is $O(\gamma_M\gamma_NMN^2+NK^3)$ because $N>M$ and $N>K$.

\subsection{Stopping Criterion}
Suppose that $K'$ targets have been detected and the residual is
\begin{align}
{\mathbf y}_r({P'})={\mathbf y} - {\sum_{l=1}^{K'}}\hat{\gamma}_{l}{\mathbf a}(\hat{p}_{l},\hat{q}_{l})
\end{align}
where $P' = \{(\hat{\gamma}_l, \hat{p}_l, \hat{q}_l), l=1,\cdots,K'\}$ denote the set of estimates of the detected targets. Then the problem is to decide whether the residual ${\mathbf y}_r({P'})$ contains the targets or not. For simplicity, a binary hypothesis testing problem where the null hypothesis ${\mathcal H}_0$ corresponds to nearly the additive white Gaussian noise (AWGN) or the alternative hypothesis ${\mathcal H}_1$ corresponds to a single target is formulated, i.e.,
\begin{align}\label{BHT}
\begin{cases}
{\mathcal H}_0:{\mathbf y}_r({P'})={\mathbf w},\\
{\mathcal H}_1:{\mathbf y}_r({P'})=\gamma{\mathbf a}({p},{q})+{\mathbf w}.
\end{cases}
\end{align}
Since this hypothesis testing problem involves unknown nuisance parameters, the GLRT can be derived. Following the procedure in \cite{Kaydetbook}, the GLRT is
\begin{align}\label{problemhard}
G_{{\mathbf y}_r(P')} = \underset{(p,q)}{\rm {max}}  \left| {\mathbf a}^{\rm H}({p},{q}){\mathbf y}_r({P'})\right|^2.
\end{align}

The stopping criterion is used to estimate the model order $K$. If the residual energy can be well explained by noise, up to a target overestimating probability, then we stop. Instead of using the exact GLRT $G_{{\mathbf y}_r(P')}$ (\ref{problemhard}), we choose to terminate the algorithm by comparing the magnitude of the whole grids of the residual with the expected noise power. The algorithm stops when
\begin{align}\label{CFAR}
G_{{\mathbf y}_r(P')}(\omega) =   \left| {\mathbf a}^{\rm H}({p},q){\mathbf y}_r({P'})\right|^2 < \tau
\end{align}
for all the sampling frequencies $p\in \Omega_p$ and $q\in \Omega_q$, where $\tau$ is the stopping threshold determined in the following text.

Supposedly, we have already correctly detected all targets in the mixture. Under this condition, the residual is ${\mathbf y}_{r}(P') \approx {\mathbf w}$, where ${\mathbf w} \sim \mathcal{CN}(\mathbf 0, \sigma^2{\mathbf I_{N}})$. Obviously, the event (\ref{CFAR}) is equivalent to
\begin{align}
\underset{(m,n)}{\operatorname{max}}|{\mathbf a}^{\rm H}(p_m,q_n){\mathbf w}|^2<\tau
\end{align}
and the overestimating probability is
\begin{align}
{\rm P}\left(\underset{(m,n)}{\operatorname{max}}|{\mathbf a}^{\rm H}(p_m,q_n){\mathbf w}|^2>\tau\right)=P_{\rm FA}.
\end{align}
Note that ${\mathbf A}^{\rm H}{\mathbf w}\sim {\mathcal {CN}}({\mathbf 0},\sigma^2{\mathbf A}^{\rm H}{\mathbf A})$, the noise ${\mathbf A}^{\rm H}{\mathbf w}$ is correlated, the threshold $\tau$ can not be obtained analytically. Therefore the threshold $\tau$ is obtained through the Monte Carlo (MC) simulations. It can be checked that $\tau=\sigma^2g({\mathbf A},M,N)$, where $g({\mathbf A},M,N)$ is the threshold with unit noise variance.

\subsection{SNR Analysis}\label{SNRana}
Let us analyze the target detection problem through SNR. Taking a single target as an example and for the $n$th received signal $r_n(t)$, define the received SNR ${\rm SNR}_{\rm r}$ as
\begin{align}
{\rm SNR}_{\rm r}=\frac{\sum\limits_{k=0}^{N_{\rm ref}-1}|r_n(kT_s)|^2}{N_{\rm ref}\sigma^2}.
\end{align}
After pulse compression, the signal is obtained as $y_n(t)$. For the sampled signal whose range cell corresponding to the true target and the received signal is perfectly matched, i.e., $\tau_l=\tau_0$, the signal is coherently integrated and the ${\rm SNR}_{\rm pc}$ is improved by $10\log N_{\rm ref}$ dB, i.e.,
\begin{align}
{\rm SNR}_{\rm pc}=10\log\frac{\sum\limits_{n=1}^N |z_n(\tau_{l})|^2}{N\sigma^2}={\rm SNR}_{\rm r}+10\log N_{\rm ref},
\end{align}
where $z_n(\tau_{l})$ is defined in (\ref{form12}). After coherent integration as shown in (\ref{defSpq}), the final ${\rm SNR}_{\rm ci}$ is improved by $10\log N$ dB, i.e.,
\begin{align}\label{CIres}
{\rm SNR}_{\rm ci}&={\rm SNR}_{\rm pc}+10\log N={\rm SNR}_{\rm r}+10\log N_{\rm ref}+10\log N.
\end{align}
For example, given that the number of pulses is $N=64$, the coherent integration gain is $10\log N\approx 18$ dB. Given the threshold $\tau$ determined by the false alarm probability, the target can be detected when the following condition
\begin{align}\label{CIcond}
{\rm SNR}_{\rm ci}&={\rm SNR}_{\rm pc}+10\log N\notag\\
&={\rm SNR}_{\rm r}+10\log N_{\rm ref}+10\log N\geq 10\log \tau
\end{align}
is satisfied, which will be validated in numerical simulations.
\section{Oversampling Factor Determination}\label{OFD}
As shown in the previous subsection, the coarse detection stage involves the determination of the oversampling factor. Obviously, increasing the  oversampling factor yields the more accurate coarse estimation and makes the Newton step effective at first. As the oversampling factor continues to increase and exceeds a certain threshold, the estimation accuracy saturates as the Newton step can always converge to the optimal solution quadratically \cite{CVXbook}.

On the other hand, the storage and computation complexity increase significantly. Thus, designing an appropriate or minimal oversampling factor to minimize the computation complexity and guaranteeing the estimation accuracy are very important for the NOMP-FAR algorithm. A necessary condition is to ensure that the Hessian matrix is positive semidefinite at the coarse estimate.

Since $S(p,q)$ depends on $\mathbf y$, which suffers from the random noise, directly analyzing $S(p,q)$ is difficult. In the following, we focus on the noiseless setting, i.e., ${\mathbf y}=\gamma_0{\mathbf a}(p_0,q_0)$, where $p_0$ and $q_0$ denote the true digital frequency. Consequently, the newly defined function $S_0(p,q,p_0,q_0)$ is obtained as
\begin{align}
S_0(p,q,p_0,q_0)\triangleq \|{\mathbf a}^{\rm H}(p_0,q_0){\mathbf a}(p,q)\|^2.
\end{align}
It is easy to see that $S_0(p,q,p_0,q_0)$ statisfy
\begin{align}
S_0(p,q,p_0,q_0)=S_0(p-p_0,q-q_0,0,0).
\end{align}
To simplify the analysis, let $p_0=q_0=0$ and define $S_0(p,q)\triangleq S_0(p,q,0,0)$, i.e.,
\begin{align}\label{S0pqdef}
S_0(p,q)&\triangleq S_0(p,q,0,0)=\|{\mathbf 1}^{\rm T}{\mathbf a}(p,q)\|^2\notag\\
&=\frac{1}{N}\sum\limits_{n=0}^{N-1}\sum\limits_{m=0}^{N-1}{\rm e}^{{\rm j}p(d_n-d_m)}{\rm e}^{{\rm j}q((1+d_n\Delta f/f_c)n-(1+d_m\Delta f/f_c)m)}.
\end{align}
Then we delve into the property of $S_0(p,q)$ to provide the guidelines of designing the algorithms.

Note that $S_0(p,q)$ (\ref{S0pqdef}) depends on $\{d_n\}_{n=0}^{N-1}$, which is often drawn randomly from a distribution. As a result, we further eliminates the effects of $\{d_n\}_{n=0}^{N-1}$ by averaging over them, and provide the guidelines of determining the oversampling factor.


At first, the average of $S_0(p,q)$ with respect to $d_n$ is obtained as
\begin{align}
	\label{bars0pq}
    \bar{S_0}(p,q) = {\rm E}[S_0(p,q)]=1+\frac{1}{N}\sum\limits_{n=0}^{N-1}\sum\limits_{m=0 m\neq n}^{N-1} {\rm e}^{{\rm j}q(n-m)}C_2[(p+qn\Delta f/f_c),-(p+qm\Delta f/f_c)],
\end{align}
where
\begin{align}
C_2(p,q) = \int_{-\infty}^{\infty}\int_{-\infty}^{\infty}p(x,y){\rm e}^{{\rm j}px+{\rm j}qy}{\rm d}x{\rm d}y
\end{align}
denotes the 2-D characteristic function of the 2-D PDF $p(x,y)$ of $d_n$ and $d_m$. When the frequency samples are statistically independent,
$\bar{S_0}(p,q)$ can be further simplified as
\begin{align}\label{S0pq}
    \bar{S_0}(p,q) = 1 + \frac{1}{N}\sum\limits_{n=0}^{N-1}\sum\limits_{m=0 m\neq n}^{N-1}  {\rm e}^{{\rm j}q(n-m)} C_1(p+qn\Delta f/f_c)  C_1(-(p+qm\Delta f/f_c)),
\end{align}
where $C_1(p) = \int_{-\infty}^{\infty} p(x){\rm e}^{{\rm j}px}{\rm d}x$.

Given that $d_n$ is uniformly drawn from the set $0,1,\cdots,M-1$, i.e., $d_n$ follows $p(x)=\frac{1}{M} \sum\limits_{l=0}^{M-1}\delta(x-l)$ and $C_1(p)$ is
\begin{align}\label{C1p}
C_1(p) = \int_{-\infty}^{\infty} p(x){\rm e}^{{\rm j}px}dx= \frac{1}{M}\frac{1-{\rm e}^{{\rm j}pM}}{1-{\rm e}^{{\rm j}p}}.
\end{align}
Substituting (\ref{C1p}) in (\ref{S0pq}), $\bar{S_0}(p,q)$ is
\begin{align}\label{barS0pq}
\bar{S_0}(p,q)= 1 +\frac{1}{N}\sum\sum\limits_{m\neq n}{\rm e}^{{\rm j}q(n-m)\left[1+\frac{(M-1)\Delta f}{2f_c}\right]} h_M\left(p+\frac{n\Delta f}{f_c}q\right) h_M\left(p+\frac{m\Delta f}{f_c}q\right),
\end{align}
where ${h}_M(x) = \frac{\sin(Mx/2)}{M\sin(x/2)}$.

In the following, the three cases corresponding to $q=0$, $p=0$, and $p\neq0$, $q\neq0$ are investigated respectively.
\subsection{$q=0$ case}\label{qzero}
For $q=0$, (\ref{barS0pq}) is simplified as
\begin{align}
\bar{S_0}(p,0)= 1 + \left(N-1\right)h_M^2\left(p\right).
\end{align}
where $p\in [-\pi,\pi)$. Define the Nyquist frequency $\Delta_{{\rm Nyq},p}$ and the normalized digital frequency $x_p$ as
\begin{align}
\Delta_{{\rm Nyq},p} &= 2\pi/M\notag\\
x_p &= p/\Delta_{{\rm Nyq},p},
\end{align}
where $x_p\in [-M/2,M/2)$.

Define
\begin{align}\label{gxdef}
g(x_p)=\frac{1}{2}h_M^2\left(2\pi x_p/M\right).
\end{align}
According to \cite{NACRC2006, NOMPTSP2016}, the Newton method converges to the solution of $g'(x_p)=0$ quadratically, if the initial guess $x_{0,p}$ lies in an interval $I_p$ around the true solution where the following conditions CI-CIII are met:
\begin{itemize}
  \item [CI] $g''(x_p)\neq0$, $\quad \forall x_p \in {I}_p$,
  \item [CII] $g'''(x_p)$ is finite, \quad $\forall x_p\in {I}_p$,
  \item [CIII] $|x_{0,p}|<1/W$, where $W_p \triangleq \sup_{x_p\in {I}}0.5\big|\frac{g'''(x_p)}{g''(x_p)}\big|$.
\end{itemize}

Here the minimum oversampling factor is investigated numerically. The simulation parameters are set as shown in Table \ref{SimPara} where $M=16$ and $N=64$. The function $g(x_p)$ (\ref{gxdef}) and its derivatives $g'(x_p)$, $g''(x_p)$, $g'''(x_p)$ in a window around the origin are evaluated and shown in Fig. \ref{gder}. In addition, the $\tilde{g}_{\mathbf d}(x_p)\triangleq (S_0(p,0)-1)/(2(N-1))$ and its derivatives corresponding to two random realization of $\{d_n\}_{n=0}^{N-1}$ are also plotted, where $d_n$ is independently and uniformly drawn from the discrete set $\{0,1,\cdots,M-1\}$. It shows that $g(x_p)$, as an average of $\tilde{g}_{\mathbf d}(x_p)$ with respect to ${\mathbf d}$ is very close to its realizations. This implies that analyzing $g(x_p)$ is enough to determine the minimum oversampling factor. The first two conditions CI and CII are met for the interval ${I}_p\subset (-1,1)$. Therefore, it only remains to satisfy the third condition. Note that for an interval $I_{0,p}$ satisfying CIII, any interval $I_p\subseteq I_{0,p}$ satisfies CIII. Thus we need to find a interval satisfying CIII as wider as possible.

To determine $I_{{\rm max},p}$, we proceed as follows: Given an interval $I_p$, the maximum $x_{{\rm max},p}=|I_p|/2$ of $x_{0,p}$ and $1/W_p$ are obtained, respectively. Then the relationship between $|I_p|$ and $x_{{\rm max},p}$, the relationship between $|I_p|$ and $1/W_p$ are obtained, respectively, as shown in Fig. \ref{IW_p}. In addition, the results corresponding to the same two random realization of $\{d_n\}_{n=0}^{N-1}$ in Fig. \ref{gder} are also plotted and match that of the averaged case well. It shows that setting $I_p$ as ${I}_p\subset (-0.28,0.28)$, then for arbitrary $x_p \in {I_p}$, one has $W_p\approx 1/0.28$ and $|x_{0,p}|<1/W_p$. In this setting, the spacing of the normalized grid should be less than $0.56$, which is equivalent to the minimal oversampling factor $1/0.56\approx 1.79$. The analysis shows that even for a single target and in a noiseless setting, the number of grids should be larger than $1.79M$.
\begin{figure}
  \centering
  \includegraphics[width=3.0in]{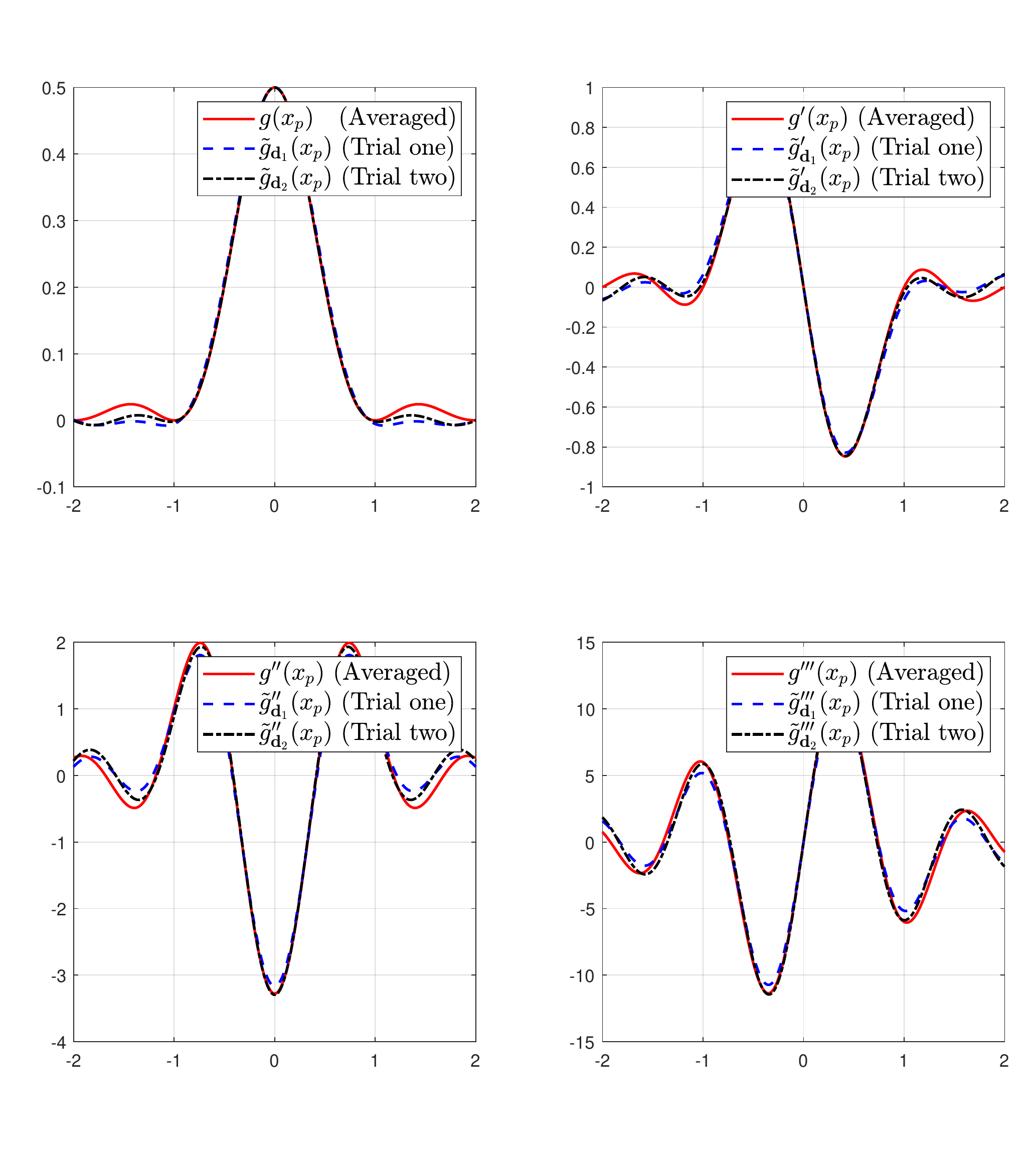}\\
  \caption{$g(x_p)$ and its derivatives. Here $M=16$ and $N=64$. In addition, the $\tilde{g}_{\mathbf d}(x_p)\triangleq (S_0(p,0)-1)/(2(N-1))$ and its derivatives corresponding to two random realization of $\{d_n\}_{n=0}^{N-1}$ are also plotted.}
  \label{gder}
\end{figure}

\begin{figure}
  \centering
  \includegraphics[width=3.0in]{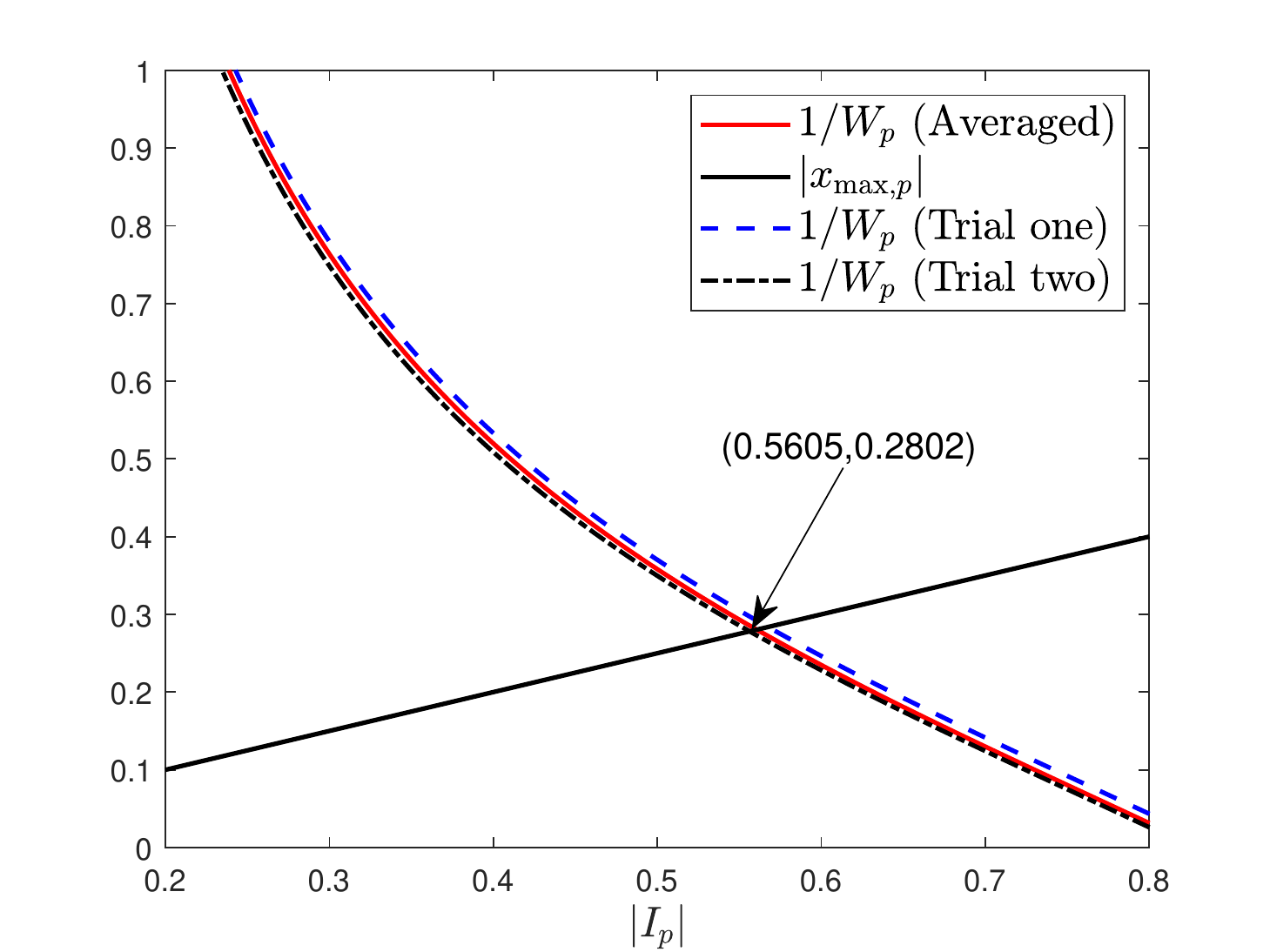}\\
        \caption{$x_{{\rm max},p}$ and $1/W_p$ versus the interval length $|I_p|$. The cross point between $x_{{\rm max},p}(|I_p|)$ and $1/W_p(|I_p|)$ is $(0.5605,0.2802)$.}
\label{IW_p}
\end{figure}

%
%
%
%

\subsection{$p=0$ Case}\label{pzero}
For $p=0$, (\ref{barS0pq}) is simplified as
\begin{align}
	\label{S0qdef}
	\bar{S}(0,q)&={\rm E}[S(p,q)|_{p=0}]= 1 + \frac{1}{N}\sum\sum\limits_{m\neq n}{\rm e}^{{\rm j}q(n-m)\left[1+\frac{(M-1)\Delta f}{2f_c}\right]} h_M\left(\frac{\Delta f}{f_c}nq\right) h_M\left(\frac{\Delta f}{f_c}mq\right).
\end{align}
Define the Nyquist frequency $\Delta_{{\rm Nyq},q}$ and the normalized digital frequency $x_q$ as
\begin{align}
\Delta_{{\rm Nyq},q} &= 2\pi/N\notag\\
x_q &= q/\Delta_{{\rm Nyq},q},\notag
\end{align}
where $x_q\in [-N/2,N/2)$. As a result, one can similarly define the function $\bar{S}(0,x_q)$ by substituting $q=\Delta_{{\rm Nyq},q}x_q$ in $\bar{S}(0,q)$.


For practical parameters, one has $\frac{\Delta f}{f_c}\ll 1$. In addition, the local behaviour around $q\approx 0$ is with great interest. Thus we approximate $h_M\left(\frac{\Delta f}{f_c}nq\right)$ and $h_M\left(\frac{\Delta f}{f_c}mq\right)$ as $h_M\left(\frac{\Delta f}{f_c}nq\right) \approx 1$, $h_M\left(\frac{\Delta f}{f_c}mq\right)\approx 1$. In this setting, ${\rm E}[S(p,q)|_{p=0}]$ can be approximated as
\begin{align}
\bar{S}_0(0,q)&={\rm E}[S(p,q)|_{p=0}] \approx Nh_N^2\left(q\left[1+\frac{(M-1)\Delta f}{2f_c}\right]\right)\triangleq \bar{S}_{0,{\rm app}}(0,q).
\end{align}
By reformulating $\bar{S}_{0,{\rm app}}(0,q)$ in terms of $x_q$, one has
\begin{align}\label{Sappqdef}
\bar{S}_{0,{\rm app}}(0,x_q) = Nh_N^2\left(\eta \frac{2\pi}{N}x_q\right).
\end{align}
where $\eta = 1+\frac{(M-1)\Delta f}{2f_c}$. Comparing with (\ref{gxdef}), $\bar{S}_{0,{\rm app}}(0,x_q)$ (\ref{Sappqdef}) involves an additional factor $\eta$ for $M=N$ (ignoring the proportional factor). Therefore, the oversampling factor of the $q$ should be $\eta$ times than that of the $p$ when $M=N$. For example, the typical parameters setting shown in Table \ref{SimPara} shows that $\eta=1.01$. As for the derivatives of $\bar{S}_{0,{\rm app}}(0,x_q)$, they similar to these of $g(x_p)$ and are omitted here.

Fig. \ref{S0andder} shows $\bar{S}_0(0,x_q)$, $\bar{S}_{0,{\rm app}}(0,x_q)$ and its derivatives. In addition, the random realizations of $\bar{S}_0(0,x_q)$ are also evaluated. It can be seen that the approximation $\bar{S}_{0,{\rm app}}(0,x_q)$ of $\bar{S}_0(0,x_q)$ is very accurate. In addition, the random realizations of ${S}_0(0,x_q)$ is close to $\bar{S}_0(0,x_q)$ averaged over $\mathbf d$. Therefore, we can analyze $\bar{S}_{0,{\rm app}}(0,x_q)$ (\ref{Sappqdef}) instead. Similar to the previous analysis of $g(x_p)$, $x_{{\rm max},q}$ and $1/W_q$ versus the interval length $|I_q|$ are evaluated and shown in Fig. \ref{IW_q} for the simulation parameter settings in Table \ref{SimPara}. It can be seen that the maximum grid spacing of the frequency should be less than $0.55$, which is equivalent to the minimal oversampling factor $1/0.55\approx1.82$.
\begin{figure}
  \centering
  \includegraphics[width=3.0in]{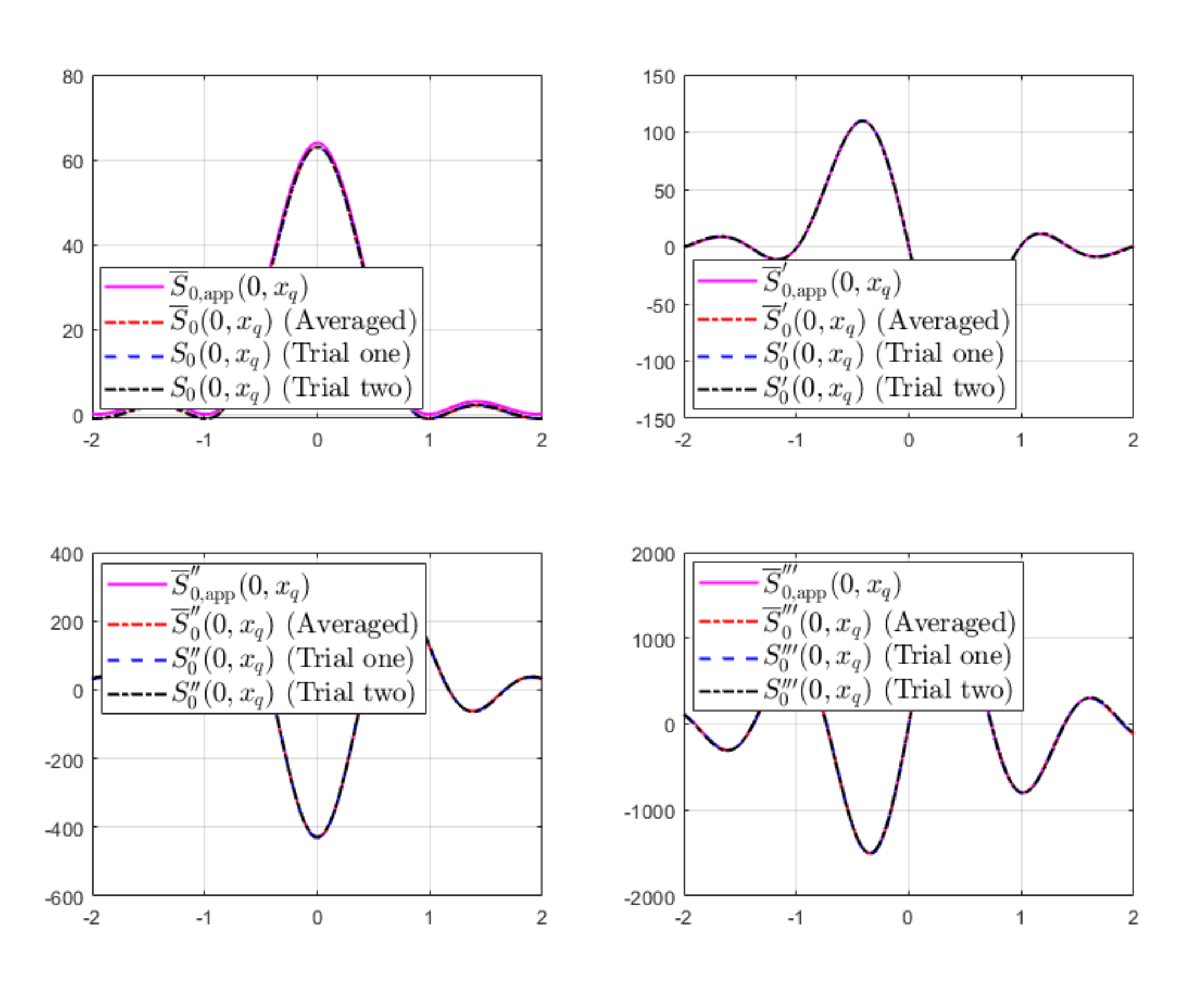}\\
  \caption{$\bar{S}_{0}(0,x_q)$, $\bar{S}_{0,{\rm app}}(0,x_q)$ and their derivatives. Here $M=16$ and $N=64$. In addition, the ${S}_{0}(0,x_q)$ and its derivatives corresponding to two random realization of $\{d_n\}_{n=0}^{N-1}$ are also plotted.}
  \label{S0andder}
\end{figure}

\begin{figure}
  \centering
  \includegraphics[width=3.0in]{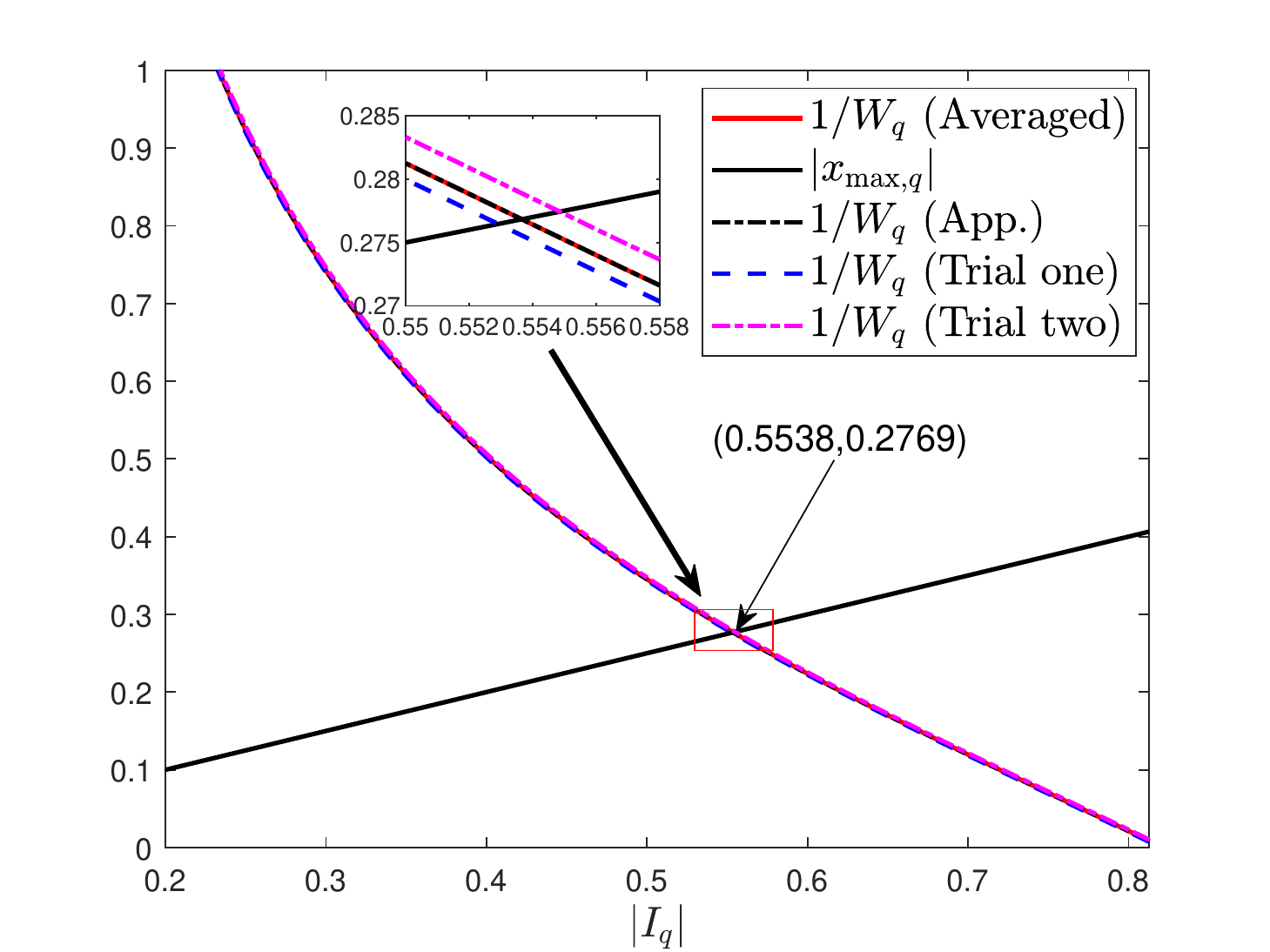}\\
  \caption{$x_{{\rm max},q}$ and $1/W_q$ versus the interval length $|I_q|$. The cross point between $x_{{\rm max},q}(|I_q|)$ and $1/W_q(|I_q|)$ is $(0.5538,0.2769)$.}
  \label{IW_q}
\end{figure}


%
\subsection{$p\neq0$ and $q\neq0$ general case}
For the general case $\bar{S}_0(p,q)$ (\ref{barS0pq}), it is hard to analyze it directly. As we are interested in the local behaviour around the origin $p=q=0$, $\bar{S_0}(p,q)$ (\ref{barS0pq}) can be simplified further. As analyzed in Subsection \ref{qzero} and \ref{pzero}, for the given simulation parameters shown in Table \ref{SimPara}, $|p|\leq 0.2802\Delta_{{\rm Nyq},p}=0.11 $ and $|q|\leq 0.2768\Delta_{{\rm Nyq},q}=0.0272$. In addition, $n\Delta f/f_c q \leq 0.085\times 0.0272=0.0023$ which is small, leading to $h_M\left(p+\frac{n\Delta f}{f_c}q\right)\approx h_M(p)$ when $0.11\geq |p|\gg 0.0023$. Therefore $\bar{S}_0(p,q)$ is approximated as
\begin{align}\label{barS0pqapp}
\bar{S_0}(p,q)\approx 1 +\frac{1}{N}\sum\sum\limits_{m\neq n}{\rm e}^{{\rm j}q(n-m)\left[1+\frac{(M-1)\Delta f}{2f_c}\right]} h_M^2\left(p\right)
=Nh_N^2\left(q\left[1+\frac{(M-1)\Delta f}{2f_c}\right]\right)h_M^2\left(p\right).
\end{align}
Now the previous analysis in Subsection \ref{qzero} and \ref{pzero} applies here. While when $|p|\leq 0.0023$, it can be checked that $h_M\left(p+\frac{n\Delta f}{f_c}q\right)\approx 1$ and Subsection \ref{pzero} applies here. Therefore, the previous analysis provides the orders of the oversampling factors. In practice, due to the random noise and intertarget interference, the oversampling factors must be set higher. Empirically, it is found that setting the oversampling factors as $\gamma_p=4$ and $\gamma_q=4$ or more works well.
\section{Full Range Bin Processing}
In practice, the goal is to recover the range of interests which is not restricted into a single range bin. The NOMP-FAR is designed for each range bin. It seems that one can recover the targets in each range bin and then the whole scene can be reconstructed. This is not so. Note that for the original model (\ref{form12}), the length of the interval where the ambiguity function $\chi(\cdot,\cdot)\neq 0$ is much longer than the sampling interval, the targets in the considered range bin can cause false alarm in the nearby range bin. To identify the false alarms, model (\ref{form12}) is analyzed in depth.

We make the following assumptions: There is only a single target in the $l$th range bin with the absolute range $r$ and velocity $v$, which corresponds to the digital frequency $p_{l}$ and $q_{l}$. Suppose that we run the NOMP-FAR for each range bin and recover a false target in the $l^{\prime}$th range bin with digital frequency $p_{l^{\prime}}$ and $q_{l^{\prime}}$. The following proposition establishes the relationship between $p_{l}$, $q_{l}$ and $p_{l^{\prime}}$, $q_{l^{\prime}}$.
\begin{proposition}
The relationship between $p_{l}$, $q_{l}$ and $p_{l^{\prime}}$, $q_{l^{\prime}}$ are: There exists an integer $h\in Z$ such that
\begin{align}\label{plreship}
p_l-2\pi \Delta f(l-l^{\prime}){{T}_{s}}=h2\pi+p_{l^{\prime}},
\end{align}
and
\begin{align}\label{qlreship}
\quad q_l=q_{l^{\prime}},
\end{align}
where $p_l,p_{l^{\prime}},q_l,q_{l^{\prime}}\in [-\pi,\pi)$.
\end{proposition}
\begin{proof}
For the $l$th range bin, according to (\ref{ptor}) and (\ref{qtov}), the relationship between the digital frequency $p_{l}$, $q_{l}$ and the range $r_l$, velocity $v_l$ are
\begin{align}\label{rl}
r_l=-\frac{c}{4\pi\Delta f}p_l+\frac{c(l-1)}{2F_s},
\end{align}
and
\begin{align}
v_l=-\frac{c}{4\pi f_c T}q_l.
\end{align}
For the $l^{\prime}$th range bin, according to (\ref{HRRPdef}), its pseudo high range resolution profile is $R_{l^{\prime}}\triangleq (\tau_0-\tau_{l^{\prime}})c/2=r_l-(l^{\prime}-1)c/(2F_s)$, here the pseudo refers to the fact that the target in that range bin is nonexist and is caused by the sidelobe of the target in the $l$th range bin. Substituting (\ref{rl}) in $R_{l^{\prime}}$, one obtains
\begin{align}\label{Rlprime}
R_{l^{\prime}}=-\frac{c}{4\pi\Delta f}p_l+\frac{c(l-l^{\prime})}{2F_s},
\end{align}
The relationship between the digital frequency $p_{l^{\prime}}$ and $R_{l^{\prime}}$ is
\begin{align}\label{Rlprimevsp}
p_{l^{\prime}}=2\pi h-\frac{4\pi \Delta f}{c}R_{l^{\prime}}.
\end{align}
Substituting (\ref{Rlprime}) in (\ref{Rlprimevsp}) and eliminating $R_{l^{\prime}}$, one obtains the desired result (\ref{plreship}). The result (\ref{qlreship}) is obvious.
\end{proof}

The false targets caused by the sidelobe must satisfy (\ref{plreship}) and (\ref{qlreship}), which can be used to eliminate the false alarms. However, there exist a scenario in which a true target is just well matched to the digital frequencies of the false alarms. Therefore, we further incorporate the amplitude constraints. Obviously, the target reflection coefficients satisfy $|\gamma_l|>|\gamma_{l^{\prime}}|$ and $|\frac{\gamma_{l^{\prime}}}{\gamma_l}|$ decreases sharply as $l$ deviates from $l^{\prime}$, which can be derived from the ambiguity function ${\rm \chi}(\xi,f_d)$ (\ref{amfdef}). Given that a target in range bin $l$ is identified as the ``true'' target via the proposed algorithm, the target obtained from the $l^{\prime}$th range bin is recognized as the false target if the following two conditions are satisfied:
\begin{itemize}
  \item [(C1)] The digital frequencies satisfy (\ref{plreship}) and (\ref{qlreship}).
  \item [(C2)]  The amplitude estimates $\hat{\gamma}_{l^{\prime}}$ and $\hat{\gamma}_l$ satisfy
  \begin{align}\label{Conds}
  20\log|\frac{\hat{\gamma}_{l^{\prime}}}{\hat{\gamma}_l}|\geq \begin{cases}
  \zeta_1,\quad |l-l^{\prime}|\leq  l_0\\
  \zeta_2,\quad {\rm otherwise},
  \end{cases}
  \end{align}
  where $l_0$, $\zeta_1$ and $\zeta_2$ are prespecified parameters.
\end{itemize}
Otherwise we regard the target obtained from the $l^{\prime}$th range bin as the ``true'' target. It is worth noting that the above two conditions only provide a way to avoid false alarms and simultaneously detect the true targets. It is heuristic and can be ineffective in extreme cases.

As a result, we propose a postprocessing algorithm termed as NOMP-FAR+Postprocessing for identify the false alarms caused by the targets in the nearby range bin, which proceeds as follows:
\begin{itemize}
  \item For each range bin $l$ where $l=l_1,\cdots,l_N$, run the NOMP-FAR to obtain the targets and sort in descending order based on the reflection amplitude $|\hat{\gamma}_{lk}|$ for each range bin. Let the reconstructed results be $\{{\mathcal P}_l\}_{l=l_1}^{l_N}$, where ${\mathcal P}_l=\{(\hat{\gamma}_{lk},\hat{p}_{lk},\hat{q}_{lk})\}_{k=1}^{\hat{K}_l}$.
  \item Recursively identify the false alarms by comparing the remaining targets with the target with the largest reflection amplitude. If both conditions (C1) and (C2) are satisfied, the target in the nearby range bin is deleted and the set ${\mathcal P}_l$ is updated. Then identify the false alarms by comparing the remaining targets with the target with the second largest reflection amplitude. Go on until all are compared.
\end{itemize}
As we show later in the numerical experiments, this method suppresses the false alarms effectively.
\section{Numerical Simulation}
\begin{table}[h!t]
\begin{center}\caption{The simulation parameters setting.}\label{SimPara}
\begin{tabular}{|c|c|}
  \hline
  Parameters & Value \\\hline
  Bandwidth $B$ & 4 MHz \\\hline
  Pulse Duration $T_p$ & 200 $\mu$s \\\hline
  Chirp Rate $\kappa$ & $2\times 10^{10}$ Hz/s \\\hline
  PRI $T$  & 1.5 ms \\\hline
  $f_c$ & 3 GHz \\\hline
  $\Delta f$ & 4 MHz \\\hline
  Sampling Frequency $F_s=B$ & 4 MHz \\\hline
  $N$ & 64 \\\hline
  $M$ & 16 \\\hline
  $R_u=c/(2\Delta f)$ & 38.5 m \\\hline
  $v_u=c/(2f_c T)$ & 33.3 m/s \\\hline
  Noise Variance $\sigma^2$ & 1 \\\hline
  Number of Measurements $N_{\rm ref}=\lfloor T_p F_s\rfloor+1$ & 801 \\
  \hline
\end{tabular}
    \end{center}
\end{table}
In this section, numerical experiments are conducted to illustrate the effectiveness of the NOMP-FAR and validate the theoretical analysis, with the simulation parameters setting in Table \ref{SimPara}.

To make performance comparison, the OMP algorithms with Nyquist sampling $\gamma_p=\gamma_q=1$ and oversampling factor $\gamma_p=\gamma_q=4$ are implemented. As for the NOMP-FAR, the parameters are set as follows: $\gamma_p=\gamma_q=4$, $R_c=3$, $R_s=20$. For the stopping criterion, the thresholds are obtained via MC simulations averaged over $10^6$ trials with unit noise variance. Unless stated otherwise, the false alarm rate $P_{\rm FA}$ is set as $10^{-2}$ and the threshold $\tau$ is $13.31$ dB for OMP ($\gamma_p=\gamma_q=4$) and NOMP-FAR ($\gamma_p=\gamma_q=4$), while for the OMP ($\gamma_p=\gamma_q=1$), the threshold $\tau$ is $10.83$ dB. Note that for a general noise variance $\sigma^2$, the threshold is $\tau\sigma^2$. For the NOMP-FAR+Postprocessing algorithm, the parameters in (\ref{Conds}) are set as follows: $l_0=3$, $\zeta_1=2$ dB, $\zeta_2=15$ dB.

Five numerical experiments are conducted. The first experiment investigates the effectiveness of the CFAR based stopping criterion. For the second experiment, the SNR improvement via coherent integration is validated. The third experiment is to show that both oversampling and Newton refinement are useful for improving the reconstruction performance. For the fourth experiment, statistical simulations are conducted to evaluate the effectiveness of NOMP-FAR in terms of hit-rate and success-rate. The hit-rate is defined as
\begin{align}
{\rm hit-rate}=\frac{1}{K}|\{k\in [1,\hat{K}]:|p_l-\hat{p}_k|\leq \delta p\quad {\rm and}~|q_l-\hat{q}_k|\leq \delta q, l\in [1,K]\}|,
\end{align}
where $\delta p$ and $\delta q$ are tolerances in range and Doppler directions, $\{p_k,q_k\}_{k=1}^K$, $\{\hat{p}_k,\hat{q}_k\}_{k=1}^{\hat{K}}$ are true digital frequencies and reconstructed digital frequencies, respectively, and $\hat{K}$ may be not equal to $K$ as the NOMP-FAR can overestimate or underestimate the number of targets. Here we set the tolerances as the half interval of the nearby grids, i.e., $\delta p=2\pi/(2\gamma_pM)=\pi/(\gamma_pM)$ and $\delta q=2\pi/(2\gamma_qN)=\pi/(\gamma_qN)$. The success-rate is defined as the empirical probability of the event in which both the model order is estimated correctly and ${\rm hit-rate}=1$. Finally, the NOMP-FAR+Postprocessing is to suppress the false targets caused by the nearby range bins and reconstruct the whole range of interest.

\begin{figure}
  \centering
  \includegraphics[width=3.0in]{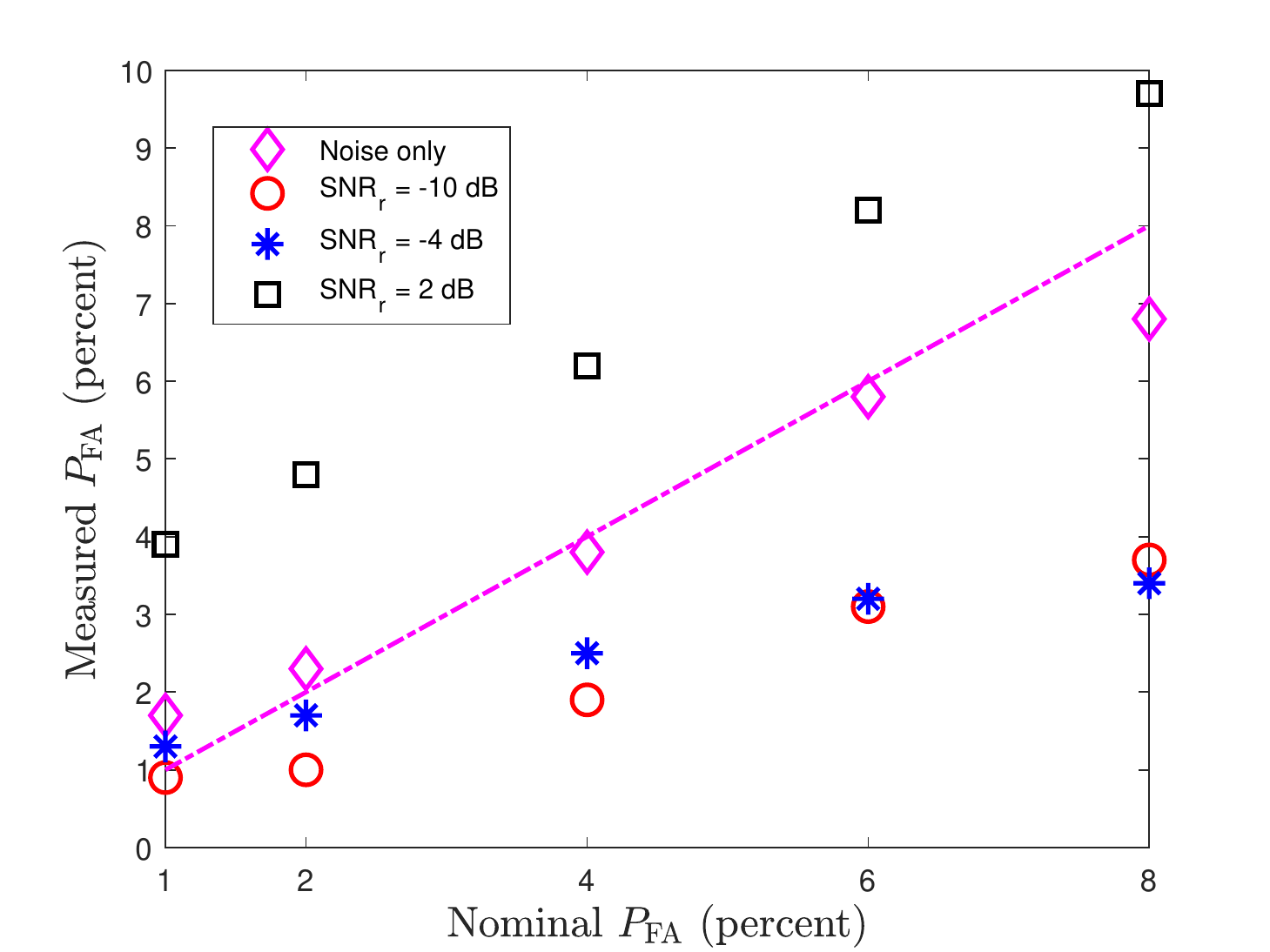}\\
  \caption{Nominal vs. measured probability of false alarm.}\label{Pfafig}
\end{figure}

\subsection{Validation of CFAR Design}

This experiment shows the effectiveness of CFAR design. Four targets are lying in the given range bin and NOMP-FAR is used to reconstruct the targets. The false alarm event is declared when NOMP-FAR overestimates the number of targets. The ``measured'' versus ``nominal'' false alarm rates for several SNRs are presented in Fig. \ref{Pfafig}. It can be seen that the empirical false alarm rates closely follow the nominal value at SNRs $-10$ dB, $-4$ dB and when the target is absent, while at a higher SNR $2$ dB, the empirical false alarm rate is always higher than the nominal value. The reason is that when SNR is high, the reconstruction accuracy of NOMP-FAR due to intertarget interference and the model approximation error where the ambiguity function ${\rm \chi}(\xi,f_d)$ in (\ref{modelam}) is dropped is not high, compared to the noise.
\begin{figure}
  \centering
  \includegraphics[width=3.0in]{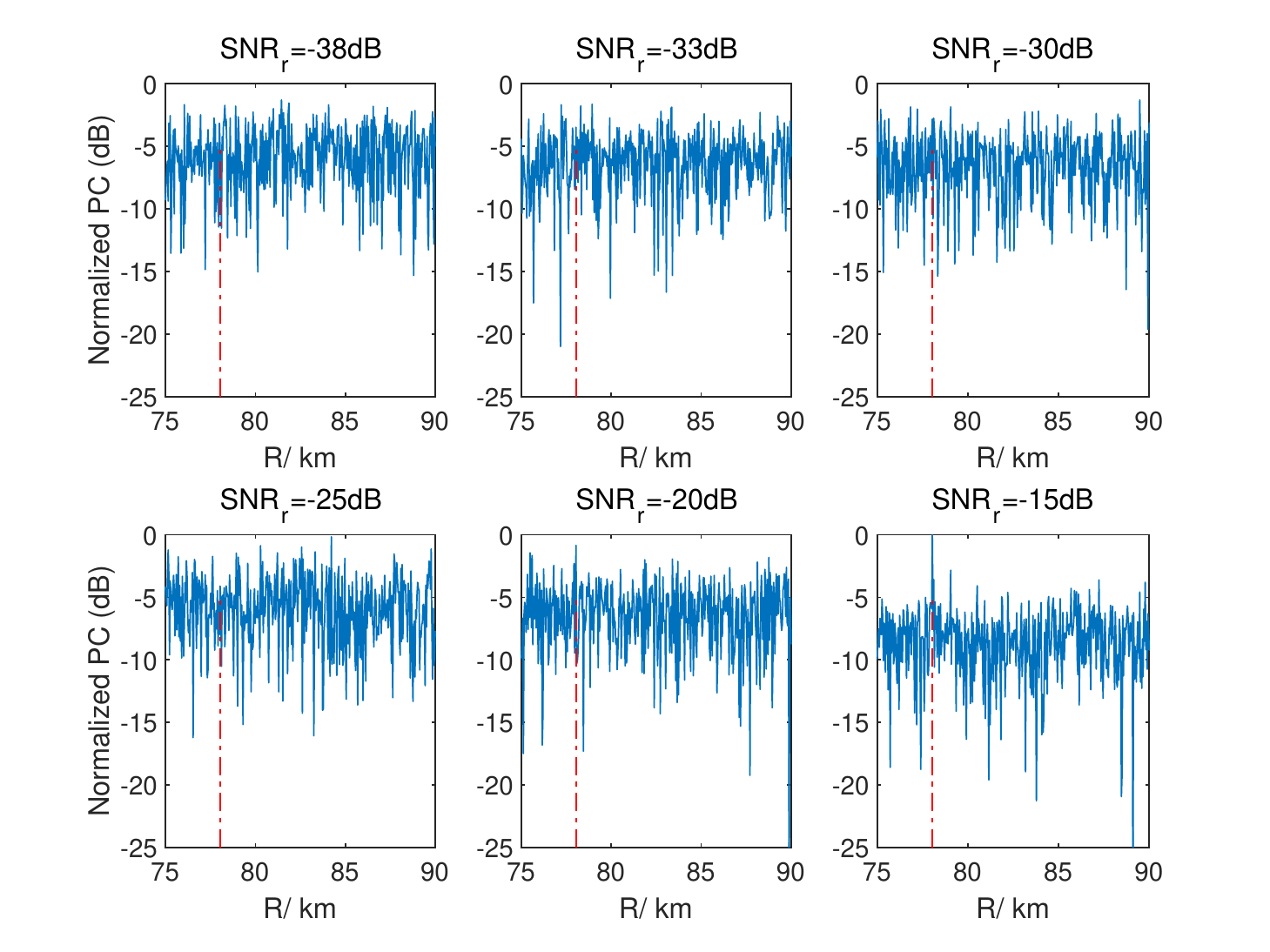}\\
  \caption{The normalized amplitude of PC at various SNRs. The dashed vertical line denotes the true target range.}\label{PCresfig}
\end{figure}
\begin{figure*}
  \centering
  \subfigure[]{
    \label{CIsig} 
    \includegraphics[width=65mm]{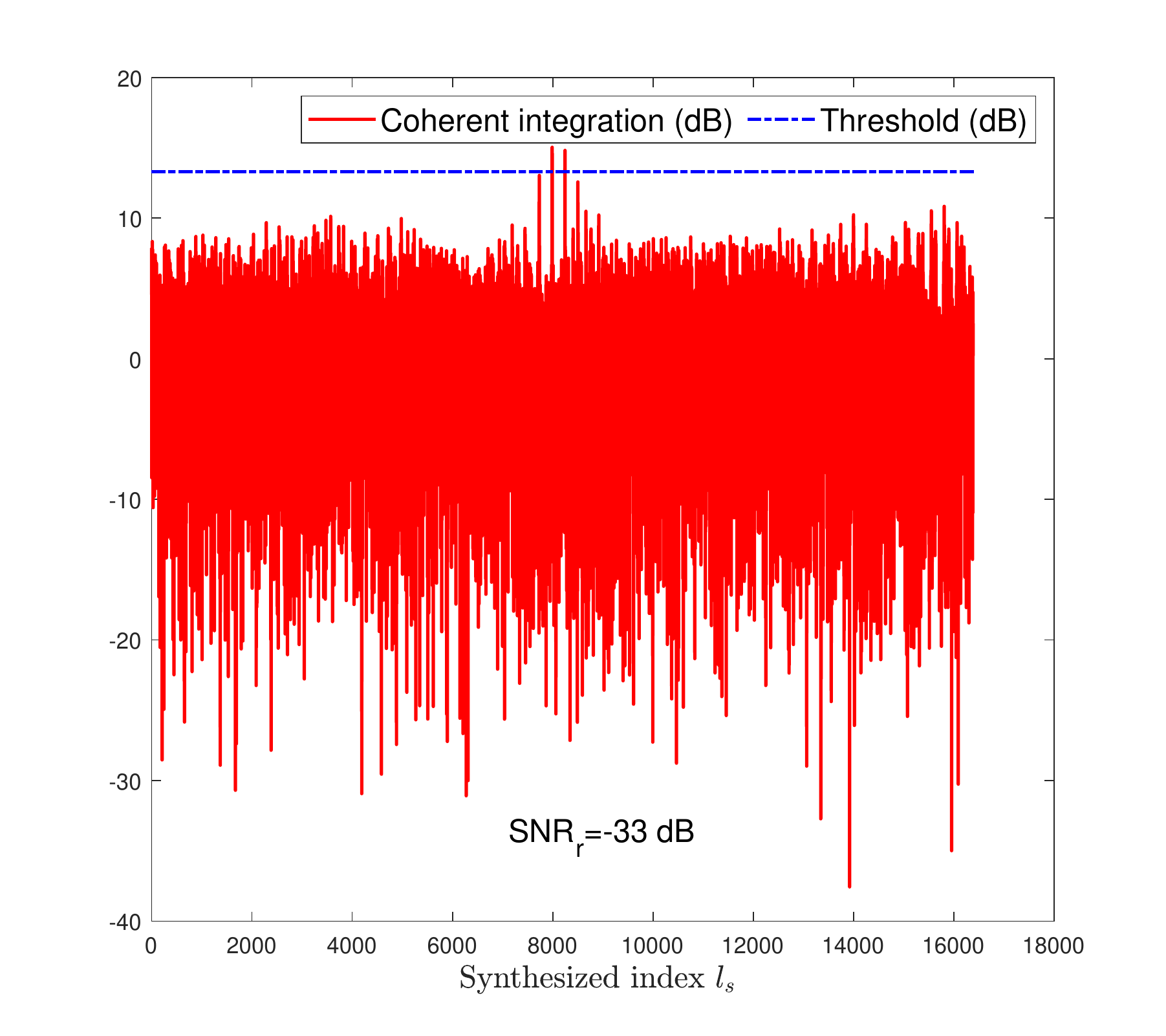}}\label{CIres}
    \subfigure[]{
    \includegraphics[width=65mm]{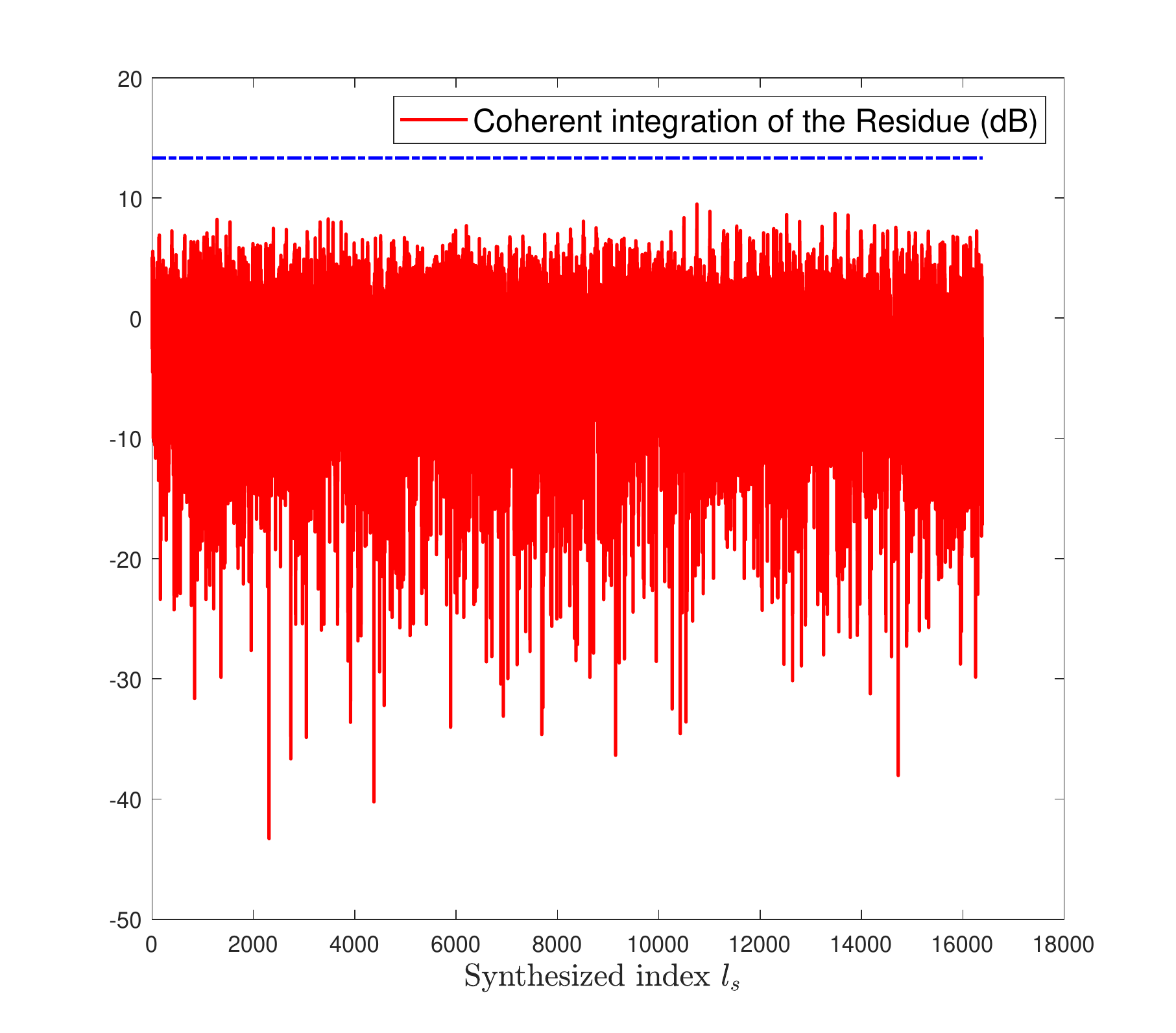}}
  \caption{The coherent integration results through $\|{\mathbf y}^{\rm H}{\mathbf a}(p,q)\|^2$ where $p\in \Omega_p$ and $q\in \Omega_q$ in a single realization. (a) Coherent integration of the signal, (b) Coherent integration of the residue after cancelling the first target.}
  \label{CIres0} 
\end{figure*}
\begin{figure}
  \centering
  \includegraphics[width=3.0in]{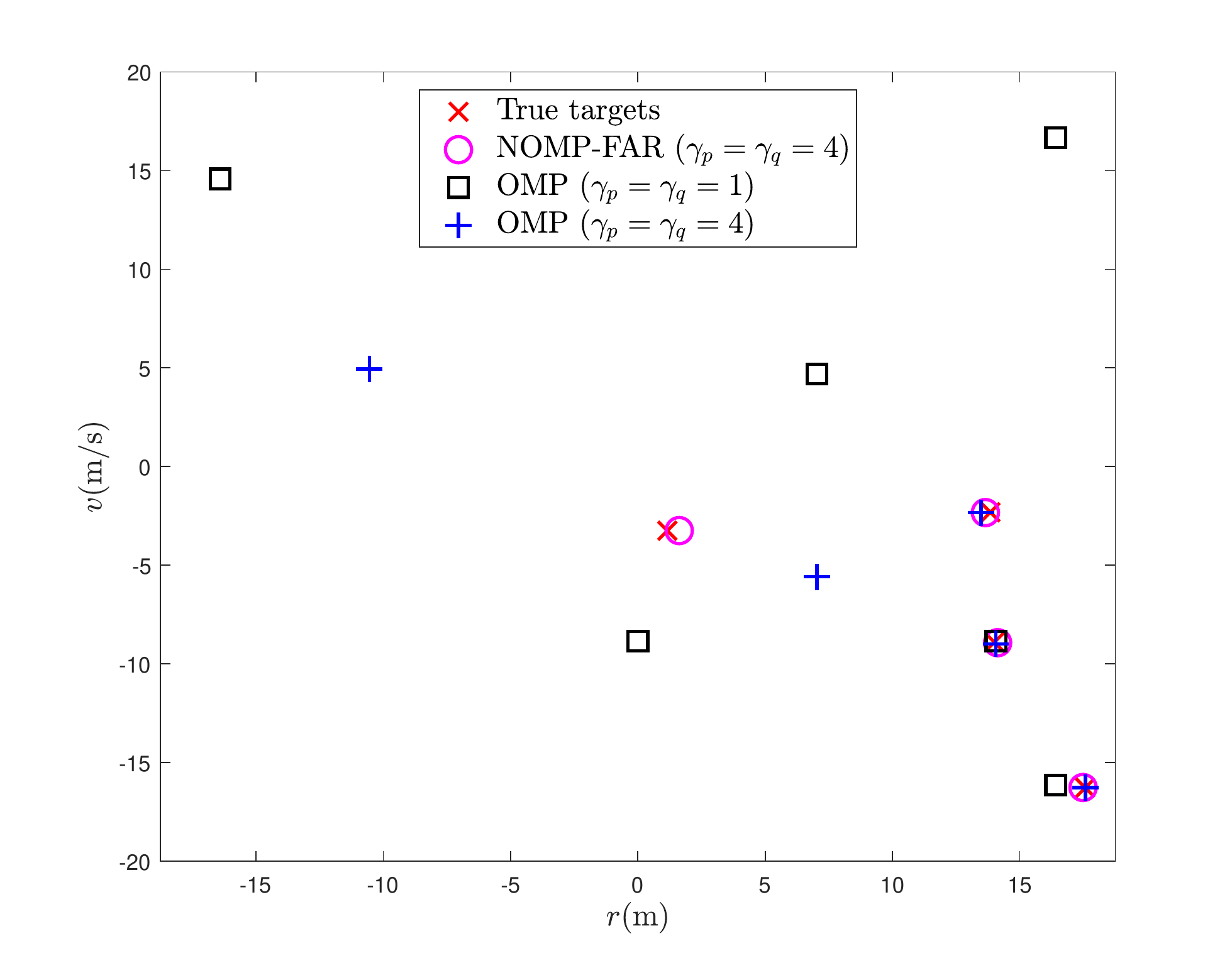}\\
  \caption{Performance comparison of NOMP-FAR and OMP in a single realization.}\label{OSandNRfig}
\end{figure}
\subsection{SNR Improvement}
This section shows the extremely low SNR under which the target can be detected. A single target locates at $78038$m with velocity $10$m/s. It can be calculated that the target lies in the $2082$ range bin exactly. Firstly, the PC results are obtained for various SNRs. For simplicity, the PC corresponding to the first pulse is plotted in Fig \ref{PCresfig} for a window limited to $75-90$ km. According to Subsection \ref{CIres}, the SNR improvement of PC is $10\log N_{\rm ref}=29.03$ dB. Therefore, the ${\rm SNR}_{\rm pc}={\rm SNR}_{\rm r}+10\log N_{\rm ref}$ are $-9,-4,-1,4,9,14$ dB for ${\rm SNR}_{\rm r}=-38,-33,-30,-25,-20,-15$ dB, respectively. It can be seen that when ${\rm SNR}_{\rm r}\geq -20$ dB, a peak appears at $78038$m corresponding to the true target localization, which can be used to identify the targets for ${\rm SNR}_{\rm r}\geq -20$ dB. The SNR improvement of coherent integration is $10\log N=18.06$ dB. Provided that the false alarm probability is setting as $P_{\rm FA}=10^{-2}$, the threshold is $13.31$ dB. From (\ref{CIcond}), when the original ${\rm SNR}_{\rm r}$ is larger than $-10\log N_{\rm ref}-10\log N+10\log\tau=-29.03-18.06+13.31=-33.78$ dB, the target can be detected. Here we set ${\rm SNR}_{\rm r}=-33$ dB. Define the synthesized index $l_{s}=k_p\gamma_qN+k_q+1$, where $k_p=0,1,\cdots,\gamma_pM-1$, $k_q=0,1,\cdots,\gamma_qN-1$. The coherent integration results versus the synthesized index $l_s$ are shown in Fig. \ref{CIres0}. It can be seen that after coherent integration, the signal level almost exceeds the threshold and the target is detected, validating the previous SNR analysis. Straightforward calculation yields the true digital frequency $p=-0.0838$ and $q=-1.8850$. The index $l_s$ corresponding to the largest peak is $l_s=7988$. Therefore $k_p=31$ and $k_q=51$, which correspond to the digital frequencies $\hat{p}=k_p2\pi/(\gamma_pM)-\pi=-0.0982$ and $\hat{q}=k_q2\pi/(\gamma_qN)-\pi= -1.8899$ and are close to the true digital frequencies $p$ and $q$. After cancelling this target, all the residue signal level are below the threshold, demonstrating that the residue is well approximated as the noise.

\begin{figure*}
  \centering
  \subfigure[]{
    \label{Hitrate} 
    \includegraphics[width=65mm]{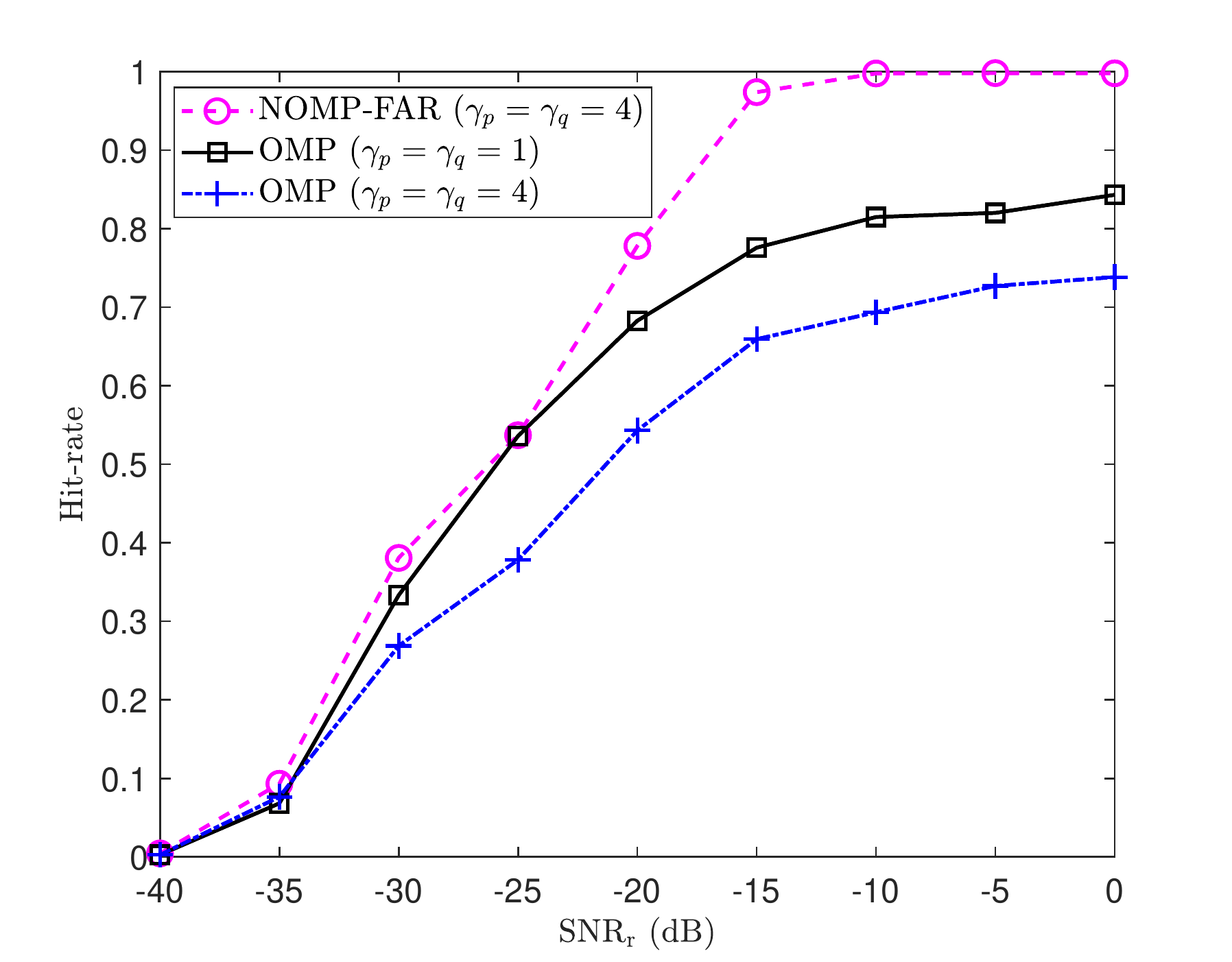}}
    \subfigure[]{
    \label{Sucs}
    \includegraphics[width=65mm]{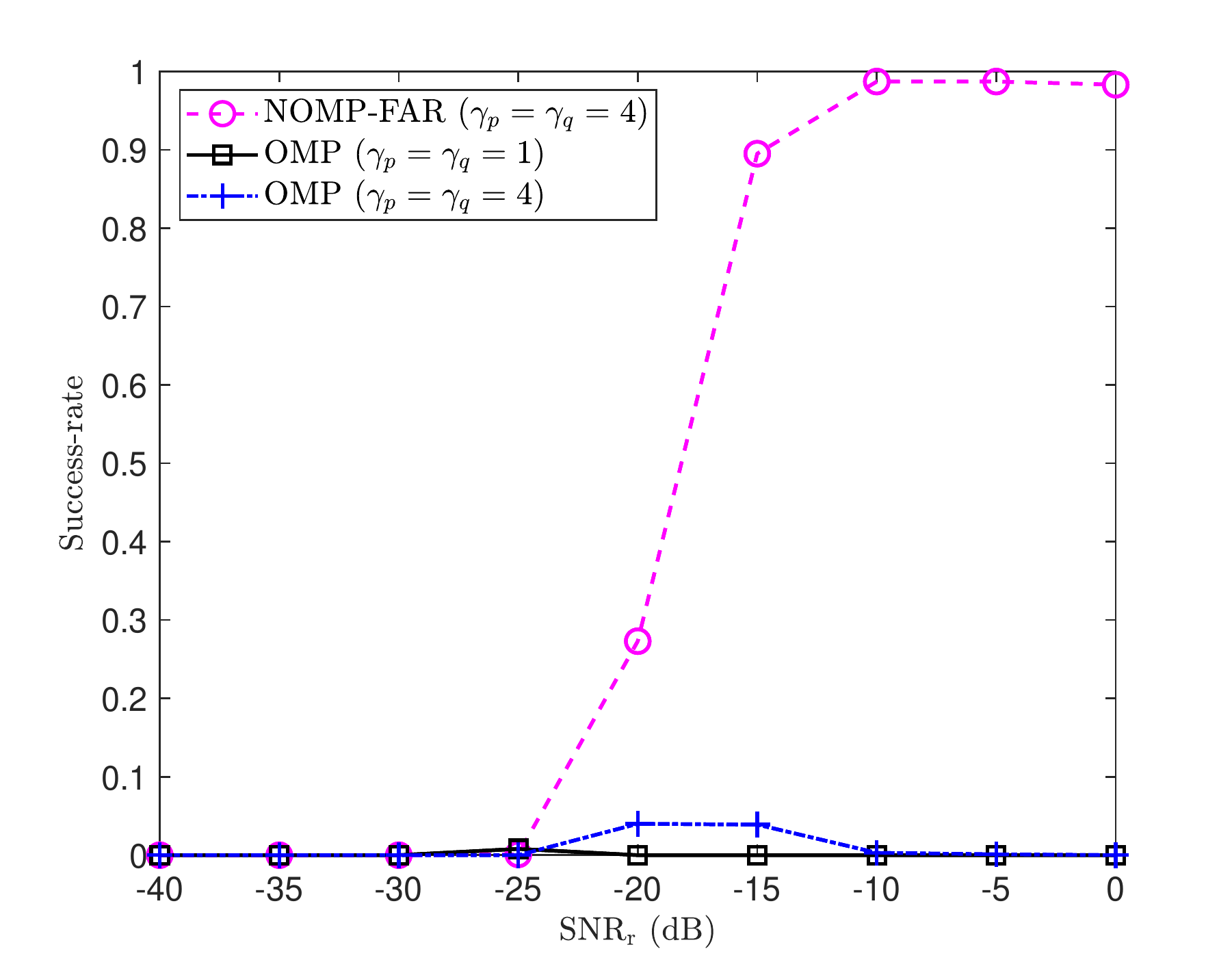}}
  \caption{Statistical performance comparison of NOMP-FAR and OMP. (a) Hit-rate versus ${\rm SNR}_{\rm r}$, (b) Success-rate versus ${\rm SNR}_{\rm r}$.}
  \label{HSMC} 
\end{figure*}
\begin{figure*}
  \centering
  \subfigure[]{
    \label{Rec15dB} 
    \includegraphics[width=65mm]{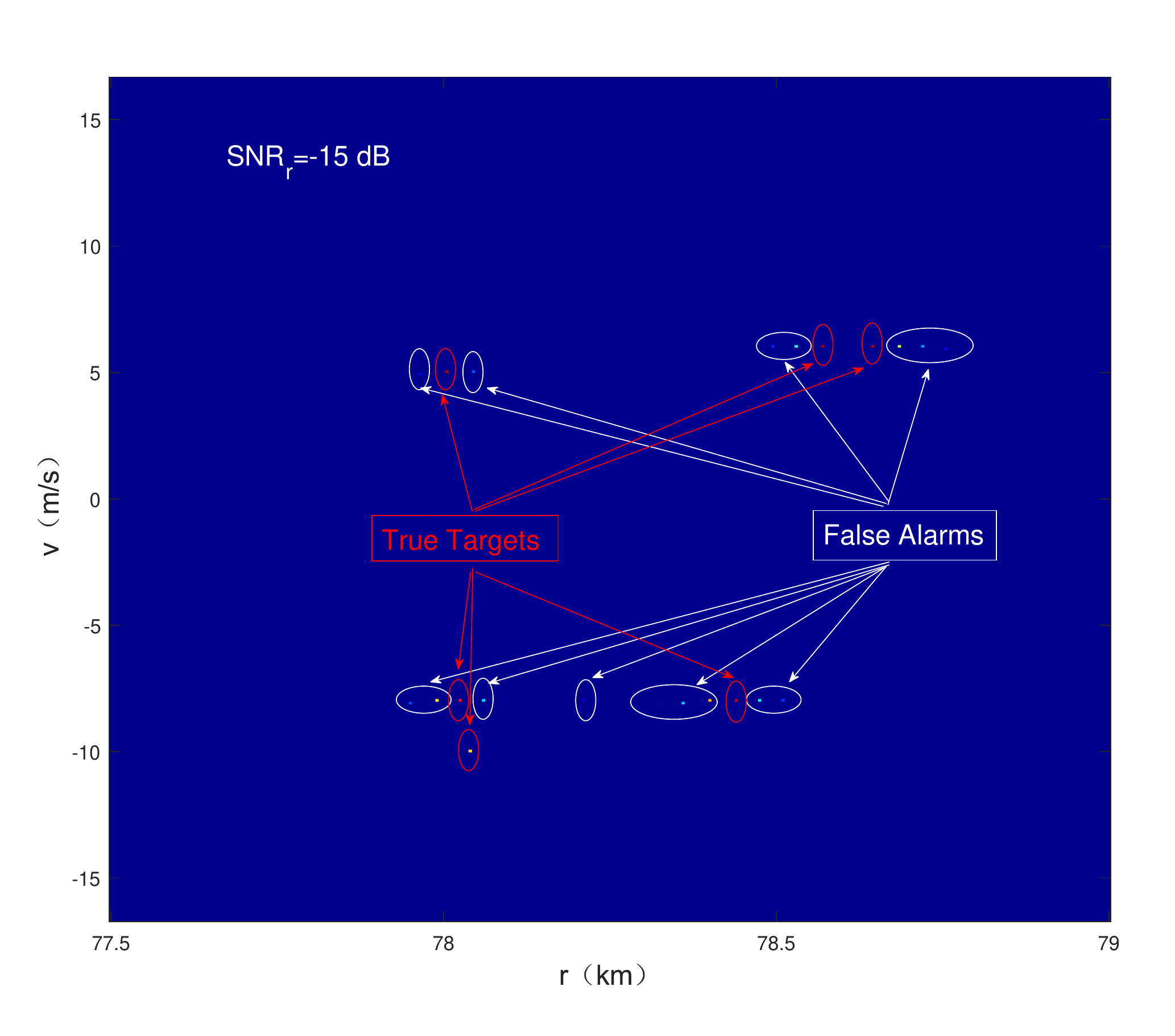}}
    \subfigure[]{
    \label{Rec15dBpost}
    \includegraphics[width=65mm]{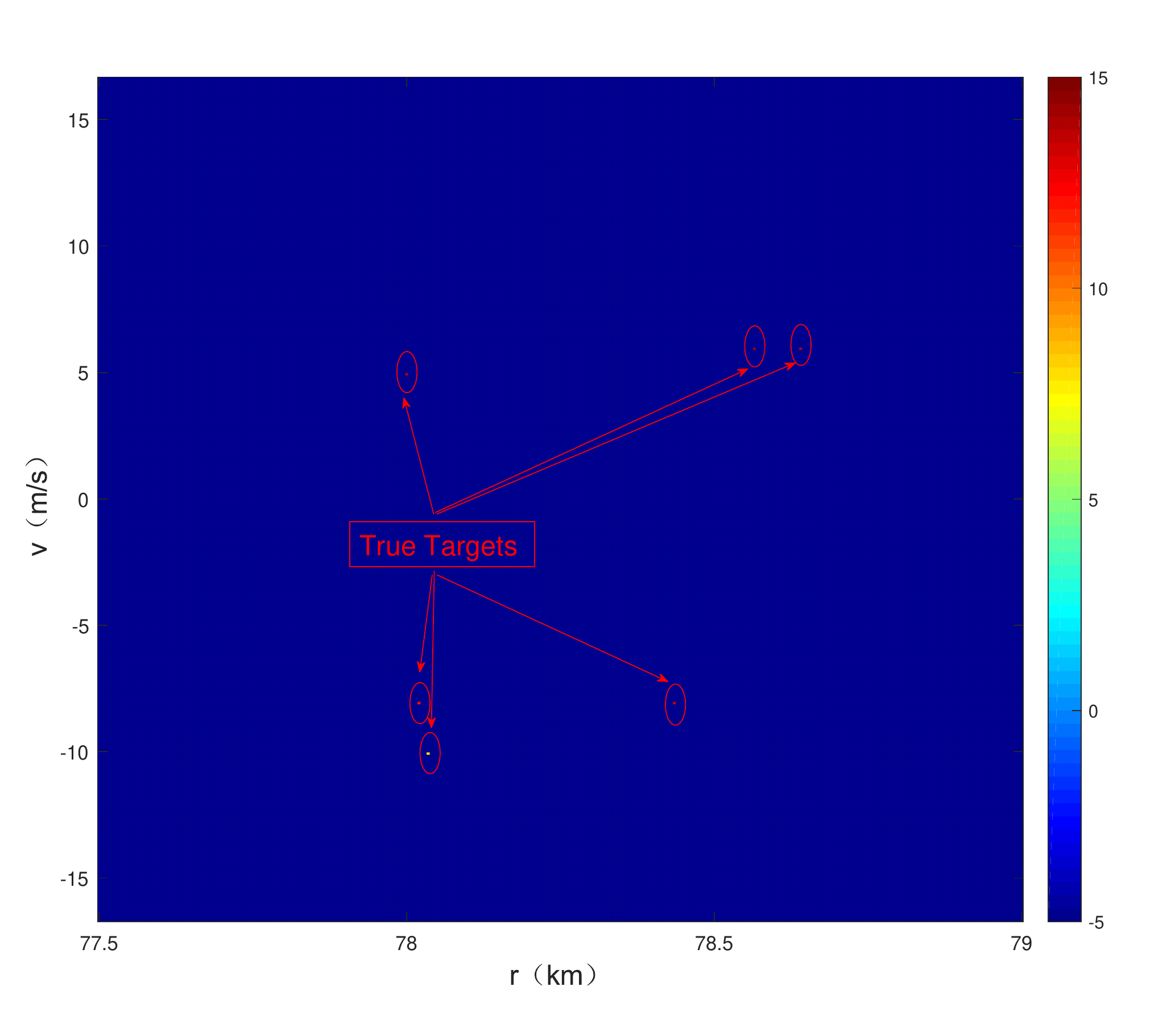}}
  \caption{The reconstruction results of NOMP-FAR and NOMP-FAR+Postprocessing at ${\rm SNR}_{\rm r}=-15$ dB in a single realization.}
  \label{Fullrangebin} 
\end{figure*}
\subsection{Effectiveness of Oversampling and Newton Refinement}\label{EON}
This experiment is conducted to demonstrate the effectiveness of oversampling and Newton refinement. Four targets are generated in the $2001$ range bin corresponding to the coarse range $75$km, and the relative ranges and velocities of the targets are uniformly and randomly drawn from the interval $[-18.75,18.75]$m and $[-16.67,16.67]$m/s, respectively. The true relative ranges and velocities of the targets are $[14,18,14,1]^{\rm T}$ m and $[-8.91,-16.29,-2.32,-3.26]^{\rm T}$m/s, respectively, and the corresponding reflection amplitudes are $[{\rm e}^{{\rm j}3.3},0.8{\rm e}^{{\rm j} 3},0.3{\rm e}^{{\rm j}3.5},0.2{\rm e}^{{\rm j}3.4}]^{\rm T}$. The SNR is ${\rm SNR}_{\rm r}=-20$ dB. The results are shown in Fig. \ref{OSandNRfig}. It can be seen that the number of targets estimated by OMP ($\gamma_p=\gamma_q=1$) and OMP ($\gamma_p=\gamma_q=4$) are $6$ and $5$, while NOMP-FAR correctly estimates the number of targets. For the estimation accuracy, the OMP ($\gamma_p=\gamma_q=1$) only estimates the first two targets with the strong reflection amplitude well, the OMP ($\gamma_p=\gamma_q=4$) performs better than OMP ($\gamma_p=\gamma_q=1$) as it estimates the first three targets with the strong reflection amplitude well. While for the NOMP-FAR, all the targets including the weakest target are estimated well. These demonstrate that oversampling and Newton refinement are beneficial for target reconstruction.
\subsection{Hit-rate and success-rate versus SNR}
The performances of NOMP-FAR and OMP versus the SNR ${\rm SNR}_{\rm r}$ are investigated, in terms of both the hit-rate and success-rate. The simulation parameters are the same as Subsection \ref{EON}. The number of MC trials is $10^3$. Results are shown in Fig. \ref{HSMC}. As shown in Fig. \ref{Hitrate}, the hit-rate increases as ${\rm SNR}_{\rm r}$ increases. For a fixed ${\rm SNR}_{\rm r}$, NOMP-FAR is the highest, followed by OMP ($\gamma_p=\gamma_q=1$) and OMP ($\gamma_p=\gamma_q=4$). The hit-rate of  OMP ($\gamma_p=\gamma_q=1$) is higher than that of OMP ($\gamma_p=\gamma_q=4$) is due to the two reasons: Firstly, the OMP ($\gamma_p=\gamma_q=1$) has estimated more number of targets than that of OMP ($\gamma_p=\gamma_q=4$). Secondly, the tolerances setting by OMP ($\gamma_p=\gamma_q=1$) is four times of that of the OMP ($\gamma_p=\gamma_q=4$). As for the success-rate shown in Fig. \ref{Sucs}, NOMP-FAR performs well, while the success-rates of the two OMPs are very low ($< 0.05$). In addition, NOMP-FAR increases as ${\rm SNR}_{\rm r}$ increases from $-40$ dB to $-10$ dB, and then begins to decrease slightly, due to the overestimating behaviour, where the recovery accuracy of NOMP-FAR and the model approximation error is large, compared to the noise level, and NOMP-FAR has to fit more number of targets to meet the stopping criterion.
\subsection{Full Range Bin Processing}
The range of interest is $[R_{\rm min},R_{\rm max}]=[77.5,79]$ km and there exists $K=6$ targets. The ranges $\mathbf r$ and velocities $\mathbf v$ of the six targets are ${\mathbf r}=[78005,78038,78025,78437.5,78570,78645]^{\rm T}$ m and ${\mathbf v}=[5,-10,-8,-8,6,6]^{\rm T}$ m/s. Note that the six targets lie in the bins ${\mathbf l}=\lfloor\frac{{\mathbf r}}{c/(2F_s)}+\frac{1}{2}\rfloor+1=[2081,2082,2082,2093,2096,2098]^{\rm T}$.
The amplitudes of the targets are $[1,0.5,-1,1.2,-1,1.2]^{\rm T}$ and the phases of the targets are drawn independently from the uniform distribution ${\mathcal U}(0,2\pi)$. Straightforward calculations show that the digital frequencies generated by the third and fourth targets satisfy (\ref{plreship}) and (\ref{qlreship}), and the range bin difference is $2093-2082=11$. While for the fifth and sixth targets, their digital frequencies also satisfy (\ref{plreship}) and (\ref{qlreship}), and the range bin difference is $2098-2096=2$. The SNR is ${\rm SNR}_{\rm r}=-15$ dB and results are shown in Fig. \ref{Fullrangebin}. It can be seen that directly running NOMP-FAR overestimates the number of targets significantly due to the sidelobes. Also, the velocities of the false targets are the same as the true target, as demonstrated in (\ref{qlreship}). After combining the postprocessing step, there is no false alarms and all the targets are detected even though the digital frequencies generated by the true targets are the same as that of the false targets. The root mean squared errors (RMSEs) of the range and velocity estimates are 0.0714m and 0.0146m/s, which demonstrate that the six targets are reconstructed with high accuracy.


\section{Conclusion}
This paper establishes the LFM FAR model and proposes a NOMP-FAR to perform HRRP and velocity estimation. The NOMP-FAR is extended to process the full range bins to avoid the false alarm caused by nearby range bins. Substantial numerical experiments are conducted to show that the proposed NOMP-FAR is able to accurately estimate the ranges and  velocities of the targets.

\end{document}